\documentclass[11pt]{elsarticle}
\usepackage[english]{babel}

\usepackage{amsmath}
\usepackage{graphicx}
\usepackage{amssymb}
\usepackage{verbatim}
\usepackage{epsfig}
\usepackage{subfig}
\usepackage{enumerate}
\usepackage{epstopdf}

\newtheorem{remark}{Remark}
\newtheorem{assumption}{Assumption}
\newtheorem{proposition}{Proposition}
\newtheorem{corollary}{Corollary}
\newtheorem{proof}{Proof}

\graphicspath{{fig/}}

\usepackage{color}

\newcommand{\col}{ \mbox{col} }
\newcommand{\rank}{ \mbox{rank } }

\def\calu{\mathcal{U}}

\def\calb{\mathcal{B}}
\def\cale{\mathcal{E}}

\def\calv{\mathcal{V}}
\def\calk{\mathcal{K}}

\def\calt{{\cal T}}

\def\hal{{1 \over 2}}

\def\L2{{\cal L}_2}
\def\L2e{{\cal L}_{2e}}

\def\rea{\mathbb{R}}

\def\diag{\mbox{diag}}

\def\begali{\begin{align}}
\def\endali{\end{align}}
\def\begequarr{\begin{eqnarray}}
\def\endequarr{\end{eqnarray}}
\def\begequarrs{\begin{eqnarray*}}
\def\endequarrs{\end{eqnarray*}}
\def\begarr{\begin{array}}
\def\endarr{\end{array}}
\def\begequ{\begin{equation}}
\def\endequ{\end{equation}}
\def\lab{\label}
\def\begdes{\begin{description}}
\def\enddes{\end{description}}
\def\begenu{\begin{enumerate}}
\def\begite{\begin{itemize}}
\def\endite{\end{itemize}}
\def\endenu{\end{enumerate}}

\def\lef[{\left[\begin{array}}
\def\rig]{\end{array}\right]}
\def\qed{\hfill$\Box \Box \Box$}
\def\begcen{\begin{center}}
\def\endcen{\end{center}}
\def\begrem{\begin{remark}\rm}
\def\endrem{\end{remark}}

\def\TAC{{\it IEEE Trans. Automat. Contr.}}
\def\AUT{{\it Automatica}}

\def\SCL{{\it Systems and Control Letters}}

\makeatletter
\def\ps@pprintTitle{%
  \let\@oddhead\@empty
  \let\@evenhead\@empty
  \def\@oddfoot{\reset@font\hfil\thepage\hfil}
  \let\@evenfoot\@oddfoot
}
\makeatother

\begin{document}
\begin{frontmatter}

\title{A Parameter Estimation Approach to  State Observation of Nonlinear Systems\tnoteref{t1}}
\tnotetext[t1]{This paper is submitted to Systems \& Control Letters Journal. The abridged version of this paper will be presented at Conference on Descision and Control $2015$, Osaka, Japan.}

 \author[LSS]{Romeo Ortega}
 \author[ITMO]{Alexey Bobtsov}
 \author[ITMO]{Anton Pyrkin}
 \author[ITMO]{Stanislav Aranovskiy}
 \address[LSS]{Laboratoire des Signaux et Syst\`emes, CNRS-SUPELEC, Plateau du Moulon, 91192, Gif-sur-Yvette, France}
 \address[ITMO]{Department of Control Systems and Informatics, ITMO University, Kronverkskiy av. 49, Saint Petersburg, 197101, Russia}

%
\begin{abstract}
A novel approach to the problem of partial state estimation of nonlinear systems is proposed. The main idea is to translate the state estimation problem into one of estimation of {\em constant, unknown parameters} related to the systems initial conditions. The class of systems for which the method is applicable is identified via two assumptions related to the transformability of the system into a suitable cascaded form and our ability to estimate the unknown parameters. The first condition involves the solvability of a partial differential equation while the second one requires some persistency of excitation--like conditions. The proposed observer is shown to be applicable to position estimation of a class of electromechanical systems, for the reconstruction of the state of power converters and for speed observation of a class of mechanical systems.
\end{abstract}


\end{frontmatter}

\section{Introduction}
\lab{sec1}
%
The problem of designing observers for nonlinear systems has received a lot of attention due to its importance in practical applications, where some of the states may not be available for measurement. The interested reader is referred to  \cite{ASTbook, BESbook} for a recent review of the literature.

In  this paper a new  framework for constructing globally convergent (reduced-order) observers for a well--defined class of nonlinear systems is presented. Instrumental to this development is to formulate the observer design problem as a problem of {\em parameter estimation}, which represents the {\em initial conditions} of the unknown part of the state. This new family of observers are called parameter estimation--based observers (PEBO). The class of systems for which PEBO is applicable is identified via two assumptions. The first one characterizes, via the solvability of a partial differential equation (PDE), systems for which there exists a partial change of coordinates that assigns a particular cascaded structure to the system that permits to obtain a classical {\em regression form} involving only measurable quantities and the unknown parameter. The second assumption pertains to our ability to consistently estimate this unknown parameter that, in general, may enter nonlinearly in the regression form. For {\em linear} regression forms, which may be obtained via over--parameterisation of the nonlinear regression, many well--established parameter estimation algorithms are available and the second assumption can be replaced by the well--known persistency of excitation (PE) condition \cite{LJUbook,SASBODbook}. The latter condition is related with the ``regular input" or ``universal input" assumptions imposed in standard observer designs \cite{ASTbook,BESbook,GAUKUPbook}.

It should be underscored that, in contrast with the classical observer design method based on  linearization up to output injection \cite{KRERES}, where the PDE to be solved imposes stringent conditions on the system, this is not the case for our PDE. The proposed PEBO also compares favourably with Kazantzis--Kravaris--Luenberger observers \cite{KAZKRA} in the following sense. Although both observers require an {\em injectivity} condition, in our observer this is imposed only on the {\em partial} change of coordinates mapping while in the Kazantzis--Kravaris--Luenberger observers the stronger requirement of injectivity of the full--state change of coordinates is needed. As is well--known  \cite{ANDPRA} ensuring the latter injectivity property is the main stumbling block for the application of this kind of observers.

The method is shown to be applicable for position estimation of a class of electromechanical systems. This class contains, as a particular case, the interesting example of permanent magnet synchronous motors (PMSM) that have been widely studied in the control and drives literature---see \cite{ACAWAT,ORTetal} and references therein. It also allows us to design observers for a class of power converters under more realistic measurement assumptions than the existing results obtained with other observer design techniques. Finally, it generates simple speed observers for mechanical systems that are partially linearisable via change of coordinates (PLvCC)---a practically important class that has been thoroughly studied in \cite{VENetal}.

The remaining of the paper is organized as follows. Section \ref{sec2} presents the problem formulation and main result.   Section \ref{sec3} is devoted to a discussion of the results. The case of linear time--invariant (LTI) systems is treated in   Section \ref{sec4}. Section \ref{sec5} illustrates the application of the technique to three physical examples. The paper is wrapped--up with concluding remarks and future research directions in  Section \ref{sec6}.
%
\section{Problem Formulation and Main Result}
\lab{sec2}
%
In this section the (partial state) observer problem addressed in the paper, and the approach that we propose to solve it, are presented.   The class of systems for which the PEBO design technique is applicable is identified via two assumptions. The first one, given in Subsection \ref{subsec22}, characterizes systems for which there exists a partial change of coordinates that assigns a particular cascaded structure to the system that permits to reformulate the state observation problem as a problem of parameter estimation. The second assumption, given in Subsection \ref{subsec23}, pertains to our ability to consistently estimate this unknown parameter.

\subsection{Partial state observer design problem}
\lab{subsec21}
%
Consider the dynamical system
\begin{align}
\nonumber
\dot x &= f_x(x,y,u)\\
\dot y &= f_y(x,y,u),
\lab{sys}
\end{align}
where $f_x: \mathbb{R}^{n_x}\times \mathbb{R}^{n_y}\times \mathbb{R}^{m} \rightarrow \mathbb{R}^{n_x}$ and $f_y: \mathbb{R}^{n_x}\times \mathbb{R}^{n_y}\times \mathbb{R}^{m} \rightarrow \mathbb{R}^{n_y}$ are smooth mappings.\footnote{Throughout the paper it is assumed that all mappings are sufficiently smooth.} Assume that the input signal vector $u: \rea_+ \to \rea^m$ is such that all trajectories of the system are bounded. Find, if possible, mappings $F: \mathbb{R}^{n_\xi}\times \mathbb{R}^{n_y}\times \mathbb{R}^{m} \rightarrow \mathbb{R}^{n_\xi}$ and $G: \mathbb{R}^{n_\xi}\times \mathbb{R}^{n_y}\times \mathbb{R}^{m} \rightarrow \mathbb{R}^{n_x}$, for some positive integer $n_\xi$, such that the (partial state) observer
\begin{align}
\nonumber
\dot \xi &= F(\xi,y,u)\\
\lab{obs}
\hat x &= G(\xi,y,u),
\end{align}
ensures that $\xi$ is bounded and
\begin{align}
\lab{obscon}
\lim_{t\rightarrow\infty}\left|\hat x(t)-x(t)\right|=0,
\end{align}
for all initial conditions $(x(0),y(0),\xi(0)) \in \rea^{n_x+n_y+n_\xi}$ and a well defined class of input signals $u \in \calu$.

It is important to underscore that, in contrast with the usual  observer problem formulation, a provision regarding the input signal is added. This additional qualifier is needed because the observation problem will be recast in terms of parameter estimation whose solution requires  ``sufficiently exciting" signals. See Section 1.2 in \cite{BESbook} for a thorough discussion of the role of the input in the observation problem. Also, note that  we have writen the system dynamics including the output $y$ as part of the state, this is done, of course, without loss of generality,
\subsection{System re--parametrization}
\lab{subsec22}
%
\begin{assumption}\em
\lab{ass1}
There exists three mappings
\begin{align*}
\phi &: \mathbb{R}^{n_x}\times \mathbb{R}^{n_y} \rightarrow \mathbb{R}^{n_z} \\
\phi^{\tt L} &: \mathbb{R}^{n_z}\times \mathbb{R}^{n_y} \rightarrow \mathbb{R}^{n_x}\\
h &: \mathbb{R}^{n_y}\times \mathbb{R}^{m} \rightarrow \mathbb{R}^{n_z},
\end{align*}
with  $n_z\ge n_x$, verifying the following conditions.
\begin{enumerate}[(i)]
\item(Left invertibility of $\phi(\cdot,\cdot)$ with respect to its first argument)
$$
\phi^{\tt L}(\phi(x,y),y)=x, \quad \forall x\in\mathbb{R}^{n_z},\;\forall y\in\mathbb{R}^{n_y}.
$$
\item(Transformability into cascade form)
\begequ
\lab{pde}
{\partial \phi \over \partial x} f_x(x,y,u)+ {\partial \phi \over \partial y}  f_y(x,y,u)=h(y,u).
\endequ
\qed
\end{enumerate}
\end{assumption}

An immediate corollary of (ii) in Assumption \ref{ass1} is that the partial change of coordinates
\begequ
\label{eq_z}
z = \phi(x,y),
\endequ
ensures
\begequ
\label{eq_zdot}
\dot z=h(y,u).
\endequ
Moreover, the left invertibility condition (i) ensures that  the partial state $x$ can be recovered from $z$ and $y$, that is,
\begequ
\label{eq_x}
x = \phi^{\tt L}(z,y).
\endequ
The cascade structure of the system is given in Fig. \ref{fig1}. 

\begin{figure}[htp]
\centering
\includegraphics[width=0.8\textwidth]{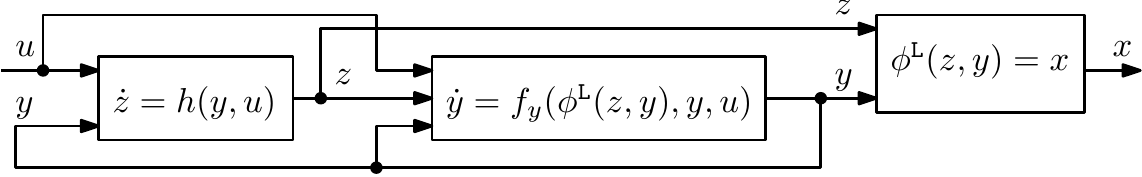}
\caption{Block diagram representation of of the transformed system.}
\label{fig1}
\end{figure}

%
\begin{proposition}\em
\lab{pro1}
Consider the system \eqref{sys} verifying Assumption \ref{ass1}. Define the dynamic extension
\begequ
\label{eq_chidot}
\dot\chi = h(y,u),
\endequ
with $\chi(0)\in \rea^{n_z}$. We can compute a mapping $\Phi: \mathbb{R}^{n_z}\times \mathbb{R}^{n_y}\times\mathbb{R}^{m}\times\mathbb{R}^{n_z} \rightarrow \mathbb{R}^{n_y}$ such that
\begin{align}
\lab{eq_ydot}
\dot y & = \Phi(\chi,y,u,\theta)\\
x & =\phi^{\tt L}(\chi+\theta,y),
\label{eq_x2}
\end{align}
where $\theta \in \rea^{n_z}$ is a vector of {\em constant}, unknown parameters.
\end{proposition}
%
%
\begin{proof}\em
From \eqref{eq_zdot} and \eqref{eq_chidot} we get $\dot z=\dot\chi.$ Hence, integrating this equation yields
\begin{align}
\label{eq_z3}
z(t)=\chi(t)+\theta,
\end{align}
where
\begin{align}
\label{eq_theta}
\theta &:= z(0)-\chi(0).
\end{align}
Replacing \eqref{eq_z3} in \eqref{eq_x} yields \eqref{eq_x2}. Finally,  the regression model \eqref{eq_ydot} is obtained replacing \eqref{eq_x2} in \eqref{sys} to get
\begin{align}
\label{eq_y2}
f_y(\phi^{\tt L}(\chi+\theta,y),y,u)=:\Phi(\chi,y,u,\theta),
\end{align}
completing the proof.
\qed
\end{proof}
\subsection{Consistent parameter estimation}
\lab{subsec23}
%
An immediate consequence of Proposition \ref{pro1} is that the problem of observation of the unmeasurable state $x$ is translated into a standard {\em parameter estimation} problem for the regression model \eqref{eq_ydot} with the observed state generated by
\begin{align}
\label{eq_hatx}
\hat x=\phi^{\tt L}(\chi+\hat \theta,y),
\end{align}
where $\hat \theta:\rea_+ \to \rea^{n_z}$ is an {\em on--line estimate} of the constant vector $\theta$. Therefore, to complete the PEBO design it is necessary to ensure the existence of a consistent estimator for the unknown parameter $\theta$.
Towards this end, the assumption below is introduced.
\begin{assumption}\em
\lab{ass2}
There exists two mappings
\begin{align*}
H &:\mathbb{R}^{n_z}\times \mathbb{R}^{n_\zeta}\times \mathbb{R}^{n_y} \times \rea^m \rightarrow \mathbb{R}^{n_\zeta} \\
N &:\mathbb{R}^{n_z}\times \mathbb{R}^{n_\zeta}\times \mathbb{R}^{n_y} \times \rea^m \rightarrow \mathbb{R}^{n_z},
\end{align*}
with  $n_\zeta > 0$ such that the parameter estimator
\begequarr
\nonumber
\dot \zeta & = & H(\chi,\zeta,y,u)\\
\lab{parest}
\hat \theta & = & N(\chi,\zeta,y,u),
\endequarr
coupled with the dynamic extension \eqref{eq_chidot} and the regression model \eqref{eq_ydot} ensures that $\zeta$ is bounded and
\begin{align}
\lab{parcon}
\lim_{t\rightarrow\infty}\left|\hat \theta(t)-\theta\right|=0,
\end{align}
for all initial conditions $(y(0),\chi(0),\zeta(0)) \in \rea^{n_y+n_z+n_\zeta}$ and a well defined class of input signals $u \in \calu$.
\qed
\end{assumption}
\subsection{Main result}
\lab{subsec24}
%
The main result of the paper is contained in the corollary below whose proof follows immediately from \eqref{eq_chidot}, \eqref{eq_x2}, \eqref{eq_hatx}, \eqref{parest} and the parameter convergence assumption \eqref{parcon}.
\begin{corollary}\em
Consider the system \eqref{sys} verifying Assumptions \ref{ass1} and \ref{ass2} with $u \in \calu$. A (partial) state PEBO of the form \eqref{obs} that guarantees \eqref{obscon} is given by
\begali
\nonumber
\xi & := \col(\chi,\zeta)\\
\nonumber
F(\xi,y,u)&:= \lef[{c} h(y,u)\\ H(\chi,\zeta,y,u)\rig]\\
\label{FG}
G(\xi,y,u)&:= \phi^{\tt L}(\chi+N(\chi,\zeta,y,u),y).
\end{align}
\qed
\end{corollary}
\subsection{On the solvability of the PDE  \eqref{pde}}
\lab{subsec25}
Similarly to all constructive observer design methods the proposed technique involves the solution of a parameterised PDE, namely \eqref{pde}---where we recall $h(y,u)$ is a {\em free} function. See  \cite{ANDPRA,ASTbook,BESbook,GAUKUPbook} for a recent review of the literature where the PDEs of the various existing observer design methods may be found. A key step in our observer design is, of course, the explicit solution of this PDE, some comments in this respect are made in this subsection.

To streamline the presentation of the subsequent discussion it is convenient to define the state vector $X:=\col(x,y)$ and the ($u$--parameterized) vector fields
\begequarrs
f_u(X) & := & \col(f_x(X,u),f_y(X,u)) \\ h_u(X) & := & h(y,u),
\endequarrs
and rewrite the PDE  \eqref{pde} element--by--element as
\begequ
\lab{pdenon}
{\partial \phi_i(X)\over \partial X} f_u(X)=(h_u)_i(X),\quad i=1,\dots,n_z.
\endequ
Following the construction used in  \cite{CHEASTORT} we then define  $n_z$ functions $\tilde \phi_i:\rea^{n_x+n_y} \times \rea \to \rea$ as
$$
\tilde \phi_i(X,s):=\phi_i(X)-s,\quad i=1,\dots,n_z,
$$
and write the non--homogeneous PDE \eqref{pdenon} as the following homogeneous PDE
\begequ
\lab{pdeele}
{\partial \tilde \phi_i(X,s)\over \partial X} f_u(X)+ {\partial \tilde \phi_i(X,s)\over \partial s} (h_u)_i(X)=0.
\endequ
Frobenius Theorem \cite{ISIbook} states that, if
\begequ
\lab{regcon}
\rank \lef[{c} f_u(X) \\ (h_u)_i(X)\rig]=1,\quad i=1,\dots,n_z,
\endequ
(uniformly in $u$) in a neighbourhood of $X_0 \in \rea^{n_x+n_y} $, then \eqref{pdeele} has a {\em local solution} around a point $(X_0,s_0) \in \rea^{n_x+n_y} \times \rea$, because all one--dimensional distributions are involutive. Unfortunately, the regularity condition \eqref{regcon}  may be a restrictive assumption in the present scenario, since rules out solutions around equilibrium points of the system \eqref{sys}. This renders Frobenius Theorem unaplicable for our problem in---often encountered---stabilisation tasks. See Remark {R5} in the next section.

Even in the case where the rank condition is satisfied it is clear that existence of solutions does not imply that an analytic expression for it can be easily obtained---and the requirement of solving the PDE remains the main stumbling block for our approach.
%
\section{Discussion}
\lab{sec3}
%
\noindent {\bf R1}  Besides the explicit solution of the PDE \eqref{pde} an additional difficulty is the selection of a mapping $h(y,u)$ that will ensure that the partial change of coordinates $\phi(x,y)$ admits a {\em left inverse} (with respect to $x$). It should be noted that this injectivity--like property is imposed on the partial change of coordinates  $\phi(x,y)$, which should be contrasted with the requirement of injectivity of the {\em full--state} change of coordinates imposed in the Kazantzis--Kravaris--Luenberger observers \cite{KAZKRA}. As discussed in \cite{ANDPRA} the latter injectivity property is the main stumbling block for the design of these observers---see \cite{PRAMARISI} for a very illustrative example.\\

\noindent {\bf R2} For general nonlinear systems the regression system \eqref{eq_ydot} depends {\em nonlinearly} on the unknown parameters. Although some results are available for the estimation of nonlinearly parameterised, nonlinear systems \cite{ANNetal,GRIetal,LIUetaltac,LIUetalscl,TYUetal} the problem of generating consistent estimates remains wide open. On the other hand,  for the case of {\em linear} parameterisation the estimation problem has a standard solution. Indeed, many techniques \cite{ASTbook,LJUbook,SASBODbook} are available to generate consistent estimates for linear regressions of the form
\begequ
\lab{linreg}
\dot y  = \Phi_0(\chi,y,u) + \Phi_1(\chi,y,u)\theta,
\endequ
with known mappings $\Phi_0:\rea^{n_z}\times \rea^{n_y} \times \rea^m \to \rea^{n_y}$ and  $\Phi_1:\rea^{n_z}\times \rea^{n_y} \times \rea^m \to \rea^{n_y \times n_z}$. See the examples in Section \ref{sec5}. \\

\noindent {\bf R3} As is well--known \cite{LJUbook}, the parameter convergence requirement in parameter estimators involves some form of excitation on the signals---this requirement is encrypted in the condition $u \in \calu$ of Assumption \ref{ass2}. This condition is of the same nature as the ``universal input" or ``regular input" conditions for classical observer designs \cite{ASTbook,BESbook,GAUKUPbook}. For linear regressions of the form \eqref{linreg} it has a very precise characterisation in terms of PE of the regressor matrix  $\Phi_1(\chi,y,u)$, which is defined as the existence of constants $\delta >0$ and $T >0$ such that for all $t \geq 0$
\begequ
\lab{pe}
\int_t^{t+T} \Phi^\top_1(\chi(s),y(s),u(s))\Phi_1(\chi(s),y(s),u(s))ds \geq \delta I_{n_z}.
\endequ
Under the PE condition above it is straightforward to design {\em globally convergent} parameter estimators for the regression model \eqref{linreg}. In this case the set $\calu$ is defined as follows:
$$
\calu:=\{u:\rea_+ \to \rea^m\;|\;\eqref{pe}\; \mbox{holds along trajectories of \eqref{eq_chidot}, \eqref{eq_ydot}}\}.
$$
From \eqref{pe} it is clear that if there are more measured states than unknown ones, that is,  if $n_y \geq n_z$, then the PE condition translates into a simple rank condition on the matrix  $\Phi_1(\chi,y,u)$---this is the case of mechanical systems treated in Subsection \ref{subsec53}. It should be also recalled that in the identification literature there are well--known relationships between the adaptation gains, the PE constants $\delta$ and $T$ and the convergence rate of the estimation errors; see \cite{LOR,SASBODbook}.\\  

\noindent {\bf R4} It should be underscored that, in many cases, it is possible to transform a nonlinearly parameterised regression into a linear one via {\em over--parameterisation}---see the discussion in this respect in \cite{LIUetalscl}. Since over--parameterisation increases the dimension of the parameter space  the excitation requirements on the signals are, of course, more stringent.\\

\noindent {\bf R5} A potential practical drawback of the proposed technique is the utilisation of pure integrators in \eqref{eq_chidot}. Indeed, in some applications the measurable signals may exhibit a (sign definite) bias in {\em steady--state} that will lead to unbounded signals when fed into open--loop integrators. On the other hand, this problem is conspicuous by its absence in regulation tasks. Indeed, in this case  there exists a desired, constant operating point $(x_*,y_*) \in \rea^{n_x+n_y}$ that must satisfy the equilibrium equations
\begin{align*}
0 &= f_x(x_*,y_*,u_*)\\
0 &= f_y(x_*,y_*,u_*),
\end{align*}
for some constant $u_* \in \rea^m$. From the equations above it is clear that a {\em necessary} condition for solvability of the PDE \eqref{pde} is that $h(y_*,u_*)=0$. Hence, in normal operating conditions, the open integration operation will not generate a bias. A similar scenario appears in the ubiquitous PI controllers widely used in industry to drive some error signal to zero. To avoid drift---{\em e.g.}, in the presence of noise---several {\em ad hoc} remedies, including the addition of small leakages and resettings, are well established.\\

\noindent {\bf R6}  As shown in  \eqref{eq_theta} the unknown parameter $\theta$---and, consequently, the estimated state---is determined by the system and observer initial conditions. It may be then argued that this  makes our analysis ``trajectory dependent", hence intrinsically fragile. This criticism would certainly be pertinent if {\em off--line}, instead of on--line, parameter estimators were advocated or if transient performance claims were made. Since this is not the case in  the present work the argument seems specious.
%
\section{Case of Linear Time--Invariant Systems}
\lab{sec4}
%
The proposed observer design procedure is of little---if at all---use for LTI systems. However, it is interesting to show that even for this simplest case the relationships between the classical notions of observability \cite{KAIbook} and identifiability \cite{WALbook} and Assumptions \ref{ass1} and  \ref{ass2} are far from obvious.
\begin{proposition}\em
\lab{pro2}
Assume the system \eqref{sys} is LTI, that is,
\begin{align}
\lab{syslti}
\left[\begin{matrix}
\dot x \\ \dot y
\end{matrix} \right]
&=
\left[\begin{matrix}
A_{11} & A_{12} \\ A_{21} & A_{22}
\end{matrix} \right]
\left[\begin{matrix}
x \\ y
\end{matrix} \right]
+
\left[\begin{matrix}
B_1 \\ B_2
\end{matrix} \right] u,
\end{align}
where $A_{ij}$ and $B_i,\;i,j=1,2$, are constant matrices of suitable dimensions.
\begenu[(C1)]
\item Observability of the system \eqref{syslti} (with respect to the output $y$) {\em does not} imply Assumption \ref{ass1}.
\item Assumption \ref{ass1} {\em does not} imply  observability of the system \eqref{syslti}.
\item If Assumption \ref{ass1} holds then observability of the system \eqref{syslti} is {\em necessary} for identifiability \cite{WALbook} of the parameter $\theta$ (defined in Proposition \ref{pro1}). Hence, it is necessary for Assumption  \ref{ass2} to hold.
\endenu
\end{proposition}
\begin{proof}\em
Without loss of generality we take
\begin{align*}
\label{rankT1}
\phi(x,y)&=T_1x+T_2y,
\end{align*}
where $T_1 \in \rea^{n_z \times n_x}$ and $T_2 \in \rea^{n_z \times n_y}$. Since $n_z \geq n_x$, condition (i) of  Assumption \ref{ass1} is ensured imposing
\begequ
\lab{rant1}
\rank T_1=n_x.
\endequ
Under this conditions we have
\begin{align}
\phi^{\tt L}(z,y)=T_1^\dag(z-T_2y),
\end{align}
where
$$
T_1^\dag:=(T_1^\top T_1)^{-1}T_1^\top .
$$
The PDE \eqref{pde} reduces to the algebraic equation
\begequ
\lab{pdelti}
T_1 A_{11}+T_2 A_{21}=0.
\endequ

Now, applying Popov--Belevitch--Hautus criterion \cite{KAIbook} we conclude that the system \eqref{syslti} is observable if and only if
\begin{align}
\label{rank2}
\rank \left[\begin{matrix}
s I_{n_x}-A_{11} \\ A_{21}
\end{matrix} \right]=n_x, \quad \forall s\in \sigma\{A_{11}\},
\end{align}
where $\sigma\{\cdot\}$ denotes the set of eigenvalues.

We prove claim (C1) by constructing a system that is observable but does not satisfy Assumption \ref{ass1}. For, take $n_x=2$ and $n_y=1$ and set $A_{11}=\diag\{a_1,a_2\}$ with $a_1 \neq a_2 \neq 0$ and $A_{21}=\lef[{cc}a & b \rig]$ with $a,b \neq 0$. The matrix appearing in \eqref{rank2} takes the form
$$
\left[\begin{matrix}
s I_{2}-A_{11} \\ A_{21}
\end{matrix} \right]=\lef[{cc}
s - a_1 & 0 \\ 0 & a - a_2 \\ a & b \rig],
$$
whose rank is two for $s=a_i,\;i=1,2$. Hence, the system is observable. On the other hand,  for all $n_z \geq 2$, \eqref{pdelti} is equivalent to
$$
T_1 =-T_2 A_{21}A^{-1}_{11}.
$$
Consequently
$$
\rank T_1 \leq \min \{\rank T_2, \rank A_{21}\}=1 < 2=n_x,
$$
violating the rank condition \eqref{rant1} .

Claim (C2) is proven by contradiction constructing a system that is not observable but satisfies Assumption \ref{ass1}. For, take $n_x=n_y=n_z=1$ and set $A_{21}=0$ that---from \eqref{rank2}---implies the system is not observable. However, if $A_{11}=0$ the algebraic equation \eqref{pdelti} admits a solution $T_1=1$, $T_2=0$ ensuring  Assumption \ref{ass1}.

Finally, claim (C3) is established proving, after some lengthy but straightforward calculations, that the regression form \eqref{eq_ydot} is given by
\begin{align*}
\dot y=A_{21}T_1^\dag\chi+(A_{22}-A_{21}T_1^\dag T_2)y+B_2u+A_{21}T_1^\dag\theta,
\end{align*}
where we notice that $A_{21}T_1^\dag \in \rea^{n_y \times n_z}$ while $A_{21} \in \rea^{n_y \times n_x}$. From  \cite{WALbook} it follows that $\theta$ is identifiable if and only if $n_y\geq n_z$ and $\rank A_{21}=n_x$---from the latter condition and \eqref{rank2} it is clear that observability follows.
\end{proof}
%
\section{Application to Three Physical Examples}
\lab{sec5}
%
In this section we prove that PEBO is applicable to the speed observation of PLvCC mechanical systems studied in \cite{VENetal}, the position observation for a class of electromechanical systems, and the reconstruction of the full state from partial measurements of a popular switched power converter.
\subsection{Mechanical systems which are partially linearisable via change of coordinates}
\lab{subsec53}
%
In this subsection we are interested in the problem of speed observation of mechanical systems described in Hamiltonian form by
\begequ
\lab{mecsys}
\lef[{c} \dot y \\ \dot x \rig]=\lef[{cc} 0 & I_s \\ - I_s & 0 \rig]\lef[{c} {\partial H \over \partial y} \\  \\  {\partial H \over \partial x}\rig]+ \lef[{c} 0 \\ G(y) \rig]u,
\endequ
where $s:={n \over 2}$ is the number of degrees of freedom of the system, $y,x \in \rea^s$ are the generalised position and momenta, respectively, $u \in \rea^m$ is the control input, $m \leq s$, $G: \rea^s \to  \rea^{s \times m}$ is the full rank input matrix. The Hamiltonian function $H:  \rea^s \times \rea^s \to \rea$ is the energy function
$$
H(y,x)=\hal x^\top M^{-1}(y)x + \calv(y),
$$
where  $M:\rea^s \to  \rea^{s \times s}$ is the positive definite inertia matrix and $\calv: \rea^s \to  \rea$ is the potential energy function. It is assumed that position $y$ is measurable and we want to estimate velocity $\dot y$ via the estimation of momenta $x$ and the relation $\dot y=M^{-1}(y) x$.

It will be shown that if the dynamics can be rendered linear in momenta (velocities) via a change of coordinates then Assumptions 1 and 2 of the paper are satisfied. Moreover, the resulting reparameterisation is linear, that is, of the form \eqref{linreg}, and the PE condition \eqref{pe} is trivially satisfied, hence the set $\calu$ is the whole input space. These systems, referred as PLvCC, have been studied in \cite{VENetal} and they have been characterised via the solvability of a PDE---see also \cite{CHA} for an intrinsic characterisation of the class. 

To present our result we need to recall the following assumption from \cite{VENetal}.
\begin{assumption}\em
\lab{ass3}
Given the inertia matrix $M(y)$. There exists a full rank matrix $\calt :\rea^{s} \rightarrow \rea^{s \times s}$ such that, for $i=1,\dots, s$,
\begequ
\lab{bii}
\calb_{(i)}(y) + \calb_{(i)}^{\top}(y)=0,
\endequ
where the matrices $\calb_{(i)}: \rea^s \to \rea^{s \times s}$ are defined as
\begin{equation}
\calb_{(i)}(y) :=  \sum_{j = 1}^{n} \Big \{[\calt_{i}, \calt_{j}]\calt_{j}^{{\top}}(M\calt \calt^{{\top}})^{-1} + \frac{1}{2}\calt_{ji}\calt\frac{\partial}{\partial q_{j}}(\calt^{\top} M\calt)^{-1}\calt^{{\top}}\Big\}, 
\label{matt}
\end{equation}
with $\calt_i(y)$ the $i$--th column of $\calt(y)$, $\calt_{ij}(y)$ its $ij$--th element and $[\calt_{i}, \calt_{j}]$ the standard Lie bracket.\footnote{A standard Lie Bracket of two vector fields $\calt_{i}(y)$, $\calt_{j}(y)$ is defined as $[\calt_{i}, \calt_{j}] := \frac{\partial \calt_{j}}{\partial y}\calt_{i} - \frac{\partial \calt_{i}}{\partial y}\calt_{j}$.} In this case the mechanical system \eqref{mecsys} is PLvCC.
\end{assumption}

\begin{proposition}\em
\lab{pro3}
Consider the mechanical system \eqref{mecsys} whose inertia matrix verifies {Assumption} \ref{ass3}. 
\begenu[(i)]
\item  {Assumption} \ref{ass1} is satisfied with the mappings
\begequarrs
\phi(x,y) & = & \calt^\top(y) x\\
\phi^{\tt L}(z,y) & = &  \calt^{-\top}(y) z\\
h(y,u) & = & - \calt^{\top}(y) \left[  {\partial \calv \over \partial y}(y)- G(y) u  \right].
\endequarrs
\item The mapping $\Phi(\chi,y,u,\theta)$ of Proposition \ref{pro1} is linear in $\theta$ and yields
\begequ
\lab{ddy}
\dot y = [ \calt^\top(y) M(y)]^{-1}(\chi + \theta). 
\endequ
\item The PE condition \eqref{pe} is satisfied for all input signals $u$. Hence, Assumption  \ref{ass2} holds with $\calu$ being the whole input space.
\endenu
\end{proposition}
\begin{proof}\em
Invoking Proposition 1 of \cite{VENetal} we have that the partial change of coordinates
$$
z =  \calt^\top(y) x,
$$
transforms the system into the form
\begequarrs
\dot y & = & [ \calt^\top(y) M(y)]^{-1}z\\
\dot z & = & - \calt^{\top}(y) \left[  {\partial \calv \over \partial y}(y)- G(y) u  \right],
\endequarrs
which establishes claims (i) and (ii).  To prove claim (iii) we refer to  \eqref{linreg} and \eqref{ddy} and identify
$$
 \Phi_1(\chi,y):= [ \calt^\top(y) M(y)]^{-1}
$$
which is a square, full rank matrix, hence \eqref{pe} is trivially satisfied. 
\end{proof}

From Proposition \ref{pro3} the design of a  globally exponentially stable momenta observer follows trivially, and is omitted for brevity. It should be underscored that the resulting observer is much simpler than the I$\&$I observer proposed in \cite{VENetal}. This is particularly true for systems that {\em do not} satisfy the integrability Assumption  2 in \cite{VENetal}, for which it is necessary to use the, rather involved, and {\em high--gain}--like technique of I$\&$I with dynamic scaling.  

\subsection{Position observation in electromechanical systems}
\lab{subsec51}
%
We consider electromechanical systems\footnote{The interested reader is referred to \cite{MEIbook,ORTetalbook} for additional details on this model.} consisting of $n_\lambda$ inductances and a single mass, whose flux and mechanical position are denoted by $\lambda \in \rea^{n_\lambda}$ and $q \in \rea$, respectively. The magnetic energy stored in the inductances is given by
$$
\cale_{\tt M}(i,q):=\hal i^\top L(q)i + \mu^\top(q) i
$$
where $L:\rea \to \rea^{n_\lambda \times n_\lambda}$ is the positive definite, (position--dependent) inductance matrix, $\mu :\rea \to \rea^{n_\lambda}$ represents the flux linkages due to the possible presence of permanent magnets and $i \in \rea^{n_\lambda}$ are the currents flowing through the inductances.    The kinetic energy of the mass is
$$
\calk(\dot q):=\hal j \dot q^2,
$$
with $j>0$ the (constant) mass inertia. We assume that the system is subject to constant external forces, {\em e.g.}, gravitational forces, but it does not have any other potential energy storing elements. Hence, the systems potential energy is given by
 $$
\calv(q):= q \tau,
$$
where $\tau \in \rea$ represents the constant external force.

The dynamical model of the system is obtained applying Euler--Lagrange equations with the energy functions above  yielding
\begequarr
\lab{el1}
L(q) {di \over dt} + L'(q) \dot q i + \mu'(q) \dot q + R i & = & B u\\
\lab{el2}
j \ddot q - \hal i^\top L'(q) i - i^\top \mu'(q) +f \dot q & = & - \tau,
\endequarr
where $(\cdot)'$ denotes differentiation,  $R = \diag\{r_1,\dots,r_{n_\lambda}\} \geq 0$ is the matrix of resistors (in series with the inductors), $B \in \rea^{n_\lambda \times m}$ is a constant input matrix, $u \in \rea^m$ are external voltage sources and $f \geq 0$ is a Coulomb friction coefficient.

As shown in \cite{ORTetalbook} the model \eqref{el1}, \eqref{el2} describes the behaviour of a large class of electromechanical systems, including the classical levitated ball and the most common electrical motors.

We are interested in the design of an observer for the mechanical position $q$ measuring $\col(i,\dot q)$, as well as the more practically interesting case when we measure only $i$. To proceed with the observer design we recall Gauss's and Ampere's laws that establish the following expression for flux linkage vector
\begequ
\lab{flu}
\lambda =  L(q)i + \mu(q).
\endequ
Moreover, Gauss's law tells us that
\begequ
\lab{dotlam}
\dot \lambda = - R i + B u.
\endequ
The latter two equations correspond, of course, to the electrical equation \eqref{el1}. Since $i$ is measurable \eqref{dotlam} proves that $L(q)i + \mu(q)$ qualifies as an admissible partial change of coordinates\footnote{Notice that, with respect to the notation  in Proposition \ref{pro1}, $\lambda$ plays the role of $z$ and $q$ is $x$.} $\phi(x,y)$ verifying condition (ii) of Assumption \ref{ass1} with $n_z=2$ and
$$
h(y,u) :=  -Ri + Bu.
$$
Condition (i) of Assumption \ref{ass1} is satisfied if---given \eqref{flu}---we can recover $q$ from measurement of $\lambda$ and $i$. Interestingly, we show below that this is the case for PMSMs.

Unfortunately, if the only measurable quantity is $i$, its time derivative ${di \over dt}$ in \eqref{el1} is not in the form of \eqref{sys}, that is $\dot y = f_y(x,y,u)$, with $x$ being only $q$ and $y$ only $i$. Indeed, besides $q$ and $i$, the electrical equation \eqref{el1} contains the velocity $\dot q$. Hence, in order to obtain the regressor from \eqref{eq_ydot} it is necessary to assume also measurement of $\dot q$.

Let us proceed now with the observer design for the (surface mount)  PMSM \cite{KRAbook,ORTetalbook}. For the sake of clarity of exposition  assume first that $i$ and $\dot q$ are measured---the requirement of measuring $\dot q$ is relaxed later. For the PMSM we have $n_\lambda=2$, $m=2$, $B=I_2$ and
\begequarr
\nonumber
L(q) & = & L I_2\\
\lab{lmu}
\mu(q) & = & \lambda_m \lef[{c} \cos(n_p q) \\ \sin(n_p q)\rig]
\endequarr
where the positive constants $L,\lambda_m$ and $n_p$ are the stator inductance, permanent magnet flux constant and number of pole pairs, respectively. Hence, defining
$$
\phi(q,i):=L i +  \lambda_m \lef[{c} \cos(n_p q) \\ \sin(n_p q)\rig]=\lambda,
$$
it is clear that the mapping
$$
\phi^{\tt L}(\lambda,i):={1\over n_p}\arctan\left(\frac{\lambda_2-L i_2}{\lambda_1-L i_1}\right),
$$
satisfies condition (i) of Assumption \ref{ass1}.

The dynamic extension \eqref{eq_chidot} is given by
\begequ
\lab{dynext}
\dot \chi = -R i + u.
\endequ
Now, from \eqref{dotlam} (with $B=I_2$) and \eqref{dynext} we have, upon integration, that
\begequ
\lab{lamchithe}
\lambda(t)=\chi(t)+\theta,
\endequ
where  $\theta:=\lambda(0)-\chi(0)$ is the unknown parameter. After some lengthy, but straightforward calculations mimicking the proof of Proposition \ref{pro1}, we obtain a linear regression form for the currents as
\begin{align}
\lab{didt}
{d \over dt} \left[\begin{matrix}
i_1 \\  { i_2 }
\end{matrix} \right]=\Phi_0(\chi,i,\dot q,u)+\Phi_1(\dot q)\theta
\end{align}
where
\begin{align*}
\Phi_0(\chi,i,\dot q,u) &:=\lef[{cc} -{R\over L} & - n_p \dot q \\ n_p \dot q & - {R\over L} \rig]\left[\begin{matrix} i_1\\ i_2 \end{matrix} \right]+n_p \dot q \lef[{c} \chi_2 \\ \chi_1 \rig] +u\\
\Phi_1(\dot q)&:= n_p \dot q \lef[{c} \chi_2 \\ \chi_1 \rig].
\end{align*}
Although the regression form \eqref{didt} can be extended with an equation for $\ddot q$ this turns out to be unnecessary to solve the parameter estimation task that consists only of two unknown parameters.

A classical parameter estimator can be designed for the linear regression form \eqref{didt}.  However, we make the important observation that the requirement of measuring $\dot q$, which is not realistic in a practical scenario, can be obviated. Towards this end we proceed as follow. First, from \eqref{flu} and \eqref{lmu} we have that
$$
|\lambda - Li|^2=\lambda_m^2.
$$
Second, replacing \eqref{lamchithe} above and expanding the square yields the {\em static} linear regression form
\begequ
\lab{newreg}
Y(\chi,i) = S^\top (\chi, i) \eta,
\endequ
where
\begequarrs
Y(\chi,i) & := & |\chi - Li|^2 \\
S (\chi, i) & := &  \lef[{c}- 2(\chi - Li) \\ 1 \rig]
\endequarrs
are, of course, measurable and the new (extended) unknown parameter is
$$
\eta:=\lef[{c} \theta \\ \lambda_m^2 - |\theta|^2 \rig].
$$
A full theoretical analysis and extensive simulations and experimental results of parameter estimators for the regressions \eqref{didt} and (a filtered version of)  \eqref{newreg} may be found in \cite{BOBPYRORT}. As shown in that paper the set $\calu$ is defined as follows:
$$
\calu:=\{u:\rea_+ \to \rea^2\;|\; \int_t^{t+T} \dot q^2(s)ds \geq \delta>0, \; \mbox{along trajectories of \eqref{el1}, \eqref{el2} and \eqref{lmu}}\}.
$$

\subsection{\'Cuk converter}
\lab{subsec52}
%
In this subsection we apply the proposed observer design technique to power converters. As will become clear below the technique applies to a broad class of converters, including the popular boost converter. For the sake of ease of exposition, instead of developing a---notationally cumbersome---general theory for a broader class of power converters, we preferred to concentrate on the specific example of  the \'Cuk power converter, depicted in Fig. \ref{fig2}.

\begin{figure}[htp]
\centering
\includegraphics[height=3.5cm]{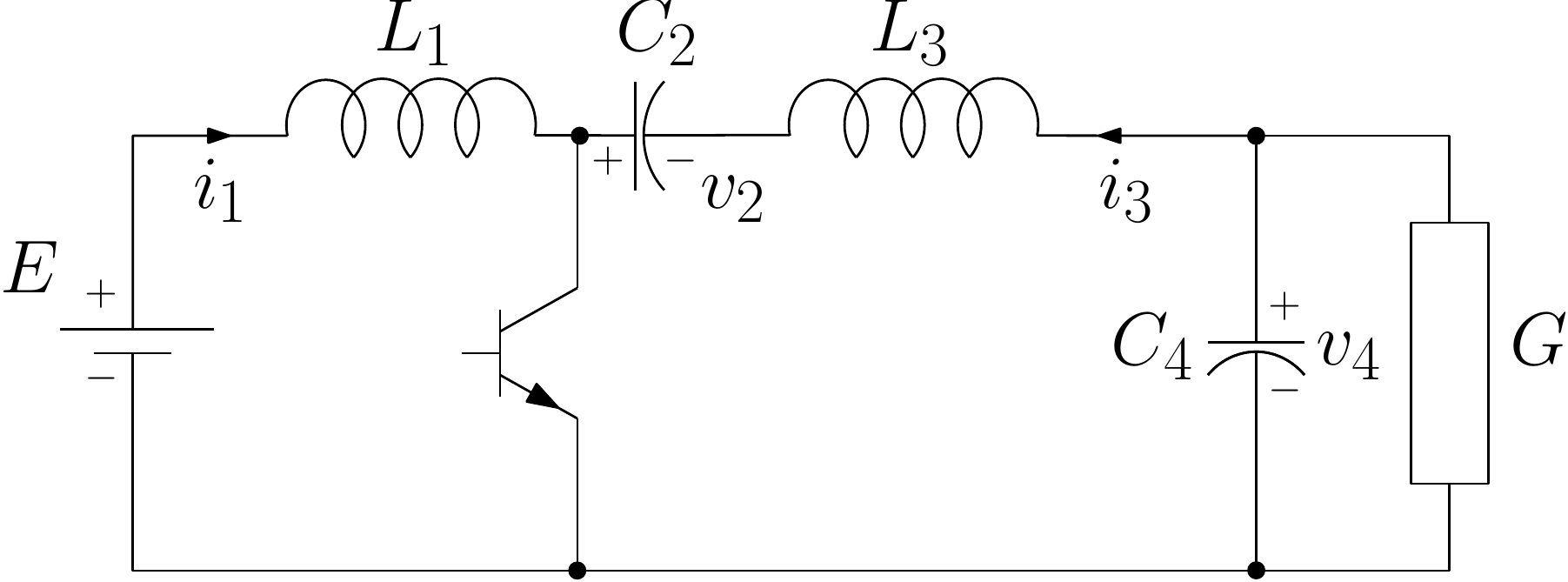}
\caption{DC--DC \'Cuk converter circuit}
\label{fig2}
\end{figure}

The average model of this device is given by
\begequ
\lab{cukdyn}
\begin{array}{lcr}
L_1{d  i_1 \over dt}=-(1-u)v_2+E\\
C_2 \dot v_2=(1-u)i_1+ui_3\\
L_3 {d  i_3 \over dt}=-uv_2-v_4\\
C_4 \dot v_4=i_3-Gv_4,
\end{array}
\endequ
where $L_1,C_2,L_3,C_4,E$ and $G$ are positive constants and $u \in (0,1)$ is a duty cycle. We refer the reader to \cite{ASTbook} for further details on the model.

To illustrate the generality of the approach we consider two different measurement scenarios. In the first one we assume that $(v_2,v_4)$ are measurable, while in the second one $(v_2,i_3)$ are measurable. Although from the practical viewpoint it is ``easier" to measure voltages, we also consider the second one since,   as shown in \cite{ASTbook}, is the one that can be solved with immersion and invariance (I$\&$I) observers, with which we compare our observer in simulations below.\\

\noindent \textbf{Case I} Denoting $y:=\col(v_2, v_4),\;x:=\col(i_1,i_3)$
we get from \eqref{cukdyn}
\begin{align*}
\dot x_1&=-{1\over L_1}(1-u)y_1+{E\over L_1}\\
\dot x_2&=-{1\over L_3}uy_1-{1\over L_3}y_2.
\end{align*}
Since the right hand side of these equations is independent of $x$ we can directly select
\begin{align*}
\phi(x,y)= x.
\end{align*}
The dynamic extension is given by
\begequarrs
\lab{chi2}
\dot\chi=\lef[{c} -{1\over L_1}(1-u)y_1+{E\over L_1}\\ -{1\over L_3}uy_1-{1\over L_3}y_2\rig]=:h(y,u),
\endequarrs
and the regression form is
\begequarr
\nonumber
\dot y & = & \Phi_0(\chi,y,u)+\Phi_1(u)\theta\\
\theta & := & x(0)-\chi(0)
\lab{regfor2}
\endequarr
where
\begequarr
\lab{case2_Phi}
 \Phi_0(\chi,y,u):=\left[\begin{matrix}
{1\over C_2}(1-u)\chi_1+{1\over C_2}u\chi_2 \\
{1\over C_4}\chi_2-{G\over C_4}y_2
\end{matrix}\right],\; \Phi_1(u):=\left[\begin{matrix}
{1\over C_2}(1-u) & {1\over C_2}u \\
0 & {1\over C_4}
\end{matrix}\right].
\endequarr
The model \eqref{regfor2} contains the time derivative of the output $y$. To get a classical (static) regression model we use the standard filtering technique \cite{MIDGOO} and define the filtered signals
$$
\overline{(\cdot)}:= {\alpha \over p + \alpha}(\cdot),
$$
where $p:={d \over dt}$ and $\alpha > 0$ is a design parameter.  Applying the filter to \eqref{regfor2} we obtain the standard linear, static regression model
\begequ
\lab{regmod}
\vartheta = \overline\Phi_1 \theta + \epsilon
\endequ
where 
$$
\lab{regmod_out}
\vartheta := {\alpha p \over p + \alpha}y- {\alpha  \over p + \alpha}\Phi_0
$$
is clearly measurable (without differentiation) and $\epsilon$ is an exponentially decaying signal that depends on the filter initial conditions and the filter time constant ${1 \over \alpha}$. 

The regression model \eqref{regmod} is used for the parameter estimator, which is the classical gradient estimator 
\begequarrs
\lab{theta_hat}
\dot{\hat\theta} & = & \Gamma \overline\Phi_1^\top (\vartheta - \overline\Phi_1\,\hat\theta),
\endequarrs
where the adaptation gain $\Gamma=\Gamma^\top>0$ is a design parameter. The state observer is defined as $\hat x=\hat\theta +\chi.$

It is important to underscore that the regressor matrix $\Phi_1(u)$  given in \eqref{case2_Phi} has an extremely simple form. Indeed, due to its upper triangular form, the estimation of the second parameter is decoupled from the first one and, moreover, the corresponding term in the regression is simply the constant ${1 \over C_4}$.  Also, since the matrix depends only on the input signal $u$ the  set $\calu$ is defined as 
$$
\calu:=\{u:\rea_+ \to (0,1)\;|\; \int_t^{t+T} \lef[{cc} 1-u(s) & (1-u(s))u(s) \\     (1-u(s))u(s)   & u^2(s) + {C^2_2 \over C^2_4} \rig] ds \geq \delta I_2>0\}.
$$
Some simple calculations show that the matrix inside the integral is {\em positive definite} for any $u \in (0,1)$. Hence,  $\calu=\{u:\rea_+ \to (0,1)\}$ and consistent estimation is always guaranteed. 

Simulations were carried out to evaluate the performance of the proposed observer.  The simulations were done for the model \eqref{cukdyn} in closed--loop with the {\em certainty equivalent} version of the full--state feedback I$\&$I controller given in Proposition 8.2 of \cite{ASTbook}. That is, the control law was defined by
\begin{align}
u={|V_d|\over |V_d|+E}+\lambda\frac{G|V_d|v_2+E(\hat x_2-\hat x_1)}{1+(G|V_d|v_2+E(\hat x_2-\hat x_1))^2}
\end{align}
where  $V_d<0$ is the {\em reference} imposed to the output voltage $v_4$ and $\lambda$ is chosen as
$$
\lambda =\lambda_0\,\min\left({|V_d|\over|V_d|+E},{E\over |V_d|+E}\right),
$$
with $0<\lambda_0<2$. The full-state version of this controller, {\em i.e.}, replacing $\hat x_1$ and $\hat x_2$ by $i_1$ and $i_3$, respectively, ensures global asymptotic stability of the desired equilibrium as well as verification of the saturation constraints in the input signal. 

The numerical simulations were performed with the following values of the converter parameters $L_1 = 10$~mH, $C_2 = 22.0$~$\mu$F, $L_3 = 10$~mH and $C_4 = 22.9$~$\mu$F, $G = 0.0447$~S and $E = 12$~V. The initial conditions for all simulations are set to $x(0)=(0.5,-1)$, $y(0)=(10 ,-12)$.  The initial set point for the output voltage is $V_d = 25$~V, and then this is changed at $t = 0.2$~s to $V_d = 30$~V, at $t = 0.4$~s to $V_d = 15$~V, at $t = 0.6$~s to $V_d = 5$~V, at $t = 0.8$~s to $V_d = 20$~V. The simulation results are presented in Fig. \ref{fig_gbo2}.\\

%
\noindent \textbf{Case II} Denoting now $y:=\col(v_2, i_3),\;x:=\col(i_1,v_4)$ we get from \eqref{cukdyn}
\begin{align*}
\dot x_1&=-{1\over L_1}(1-u)y_1+{E\over L_1}\\
\dot x_2&={1\over C_4}y_2-{G\over C_4}x_2.
\end{align*}
The right hand side of the second equation depends on $x_2$, therefore the choice $\phi(x,y)=x$ is not suitable here. We propose instead
\begin{align*}
\phi(x,y)= x- \lef[{c} 0 \\ {GL_3\over C_4}y_2\rig],
\end{align*}
that, introducing the partial change of coordinates $z=\phi(x,y)$, yields the required form
$$
\dot z =\lef[{c} -{1\over L_1}(1-u)y_1+{E\over L_1}\\ {1\over C_4}y_2+{G\over C_4}uy_1\rig]=:h(y,u).
$$
The dynamic extension is then given by $\dot\chi=h(y,u),$ and the regression model is of the form 
\begequarrs
\nonumber
\dot y & = & \Phi_0(\chi,y,u)+\Phi_1(u)\theta\\
\theta & := & x(0)-\chi(0)-\lef[{c} 0 \\ {GL_3\over C_4}y_2(0)\rig],
\lab{regfor1}
\endequarrs
where
\begequ
\lab{phi1}
 \Phi_0(\chi,y,u):=\left[\begin{matrix}{1\over C_2}(1-u)\chi_1+{1\over C_2}uy_2 \\
-{1\over L_3}uy_1-{1 \over L_3}\chi_2 -{GL_3\over C_4}y_2\end{matrix}\right],\;
\Phi_1(u):=\left[\begin{matrix} {1\over C_2}(1-u) & 0 \\ 0 & -{1\over L_3} \end{matrix}\right].
\endequ
The state observer takes the form
\begequarrs
\lab{x_hat}
\hat x=\hat\theta +\chi+\lef[{c} 0 \\ {GL_3\over C_4}y_2\rig].
\endequarrs

The regressor matrix $\Phi_1(u)$  given in \eqref{phi1} has an even simpler form than the one of Case I above. Indeed, the matrix is now diagonal with the second term in the regression simply the constant ${-1 \over L_3}$.  Clearly, for this case we also have  $\calu=\{u:\rea_+ \to (0,1)\}$ and consistent estimation is always guaranteed. 

In Proposition 8.3 of \cite{ASTbook} the following I$\&$I observer is proposed
\begequarr
\nonumber
\hat x_{I\&I} &=& \lef[{c} \zeta_1 \\ \zeta_2 \rig]+\lef[{c} C_2\gamma_1y_1 \\ L_3\gamma_2y_2\rig]\\
\nonumber
\dot{\hat \zeta}_1&=&{1\over L_1}[ -(1-u)y_1+E]-\gamma_1[(1-u)(\hat \zeta_1+C_2\gamma_1 y_1)+uy_2]\\
\lab{I&I}
\dot{\hat \zeta}_2&=& {1\over C_4} [y_2-G(\hat \zeta_2-L_3\gamma_2y_2)]-\gamma_2[uy_1+\hat \zeta_2-L_3\gamma_2y_2],
\endequarr
where $\gamma_1,\gamma_2 > 0$ are design parameters. It should be noted that in the latter reference the parameters $E$ and $G$ are treated as unknown and are also estimated. If they are assumed known the I$\&$I observer takes the form given above.

The performance of our observer was compared with the I$\&$I observer \eqref{I&I} via numerical simulations. They were done under the same scenario as the ones done for Case I, but now with the certainty equivalent observer that results replacing  $i_1$ and $v_4$ by $\hat x_1$ and $\hat x_2$, respectively. The simulation results are presented in Figs. \ref{fig_gbo1}, \ref{fig_gbo1_2} with different observer gains. 

%
\section{Conclusions}
\lab{sec6}
%
A radically new approach to design state observers for nonlinear systems has been proposed. The key idea is to translate the state observation problem into one of parameter estimation. It turns out that this is possible if we can find a partial change of coordinates $z=\phi(x,y)$ such that $\dot z$ depends only on $y$ and $u$. The observer then comprises a copy of $\dot z$ (called $\dot \chi$) that, upon integration, differs from $z$ only on the initial conditions---and these are the (constant) parameters that we propose to estimate. If the change of coordinates satisfies an injectivity property then $x$ can be estimated  from the knowledge of $y,u,\chi$ and an estimate of $\theta$. Clearly, if the latter converges to $\theta$, then the estimate of $x$ will converge to its true value.

It has been shown in the paper that the change of coordinates is obtained from the solution of a parameterised PDE, which does not impose ({\em a priori}) the strict constraints of the classical observer design \cite{KRERES}. Moreover, it is argued that the required injectivity property is weaker than the one required in the Kazantzis--Kravaris--Luenberger and the I$\&$I observers.

The design of the observer is completed adding a parameter estimator to a regression model of the form $\dot y  = \Phi(\chi,y,u,\theta)$ that, in general, depends nonlinearly on the parameters. Although some estimation techniques for nonlinearly parameterised nonlinear systems are available, it is also suggested that---via over--parameterisation---it may be possible to transform the regression into a linearly parameterised one. The latter case has been widely studied in the literature and many algorithms that guarantee parameter convergence under some excitation conditions are available.

The proposed technique has been shown to be applicable to position estimation of a class of electromechanical systems, to power converters and to speed estimation of the PLvCC mechanical systems studied in \cite{VENetal}.

Current research is under way in the following directions.
\begite
\item Identify other classes of physical systems to which the proposed method is applicable. 
\item Compare the performance of PEBO for the PLvCC mechanical systems discussed in Subsection \ref{subsec53} with other existing speed observers.
\item Further clarify under which conditions the required partial change of coordinates exists and when it will lead to an easily tractable, linearly parameterised,  regression model.
\item Further explore the connection between  the classical concepts of observability and identifiability and Assumption \ref{ass1} and the excitation conditions required by the method. In this respect, the analysis of the simplest LTI case of Section \ref{sec4}  shows that this task is far from obvious. 
\endite

{
\begin{figure*}[htp]
\centering
\subfloat[][]{{\label{fig_3a}}\includegraphics[width=0.43\textwidth]{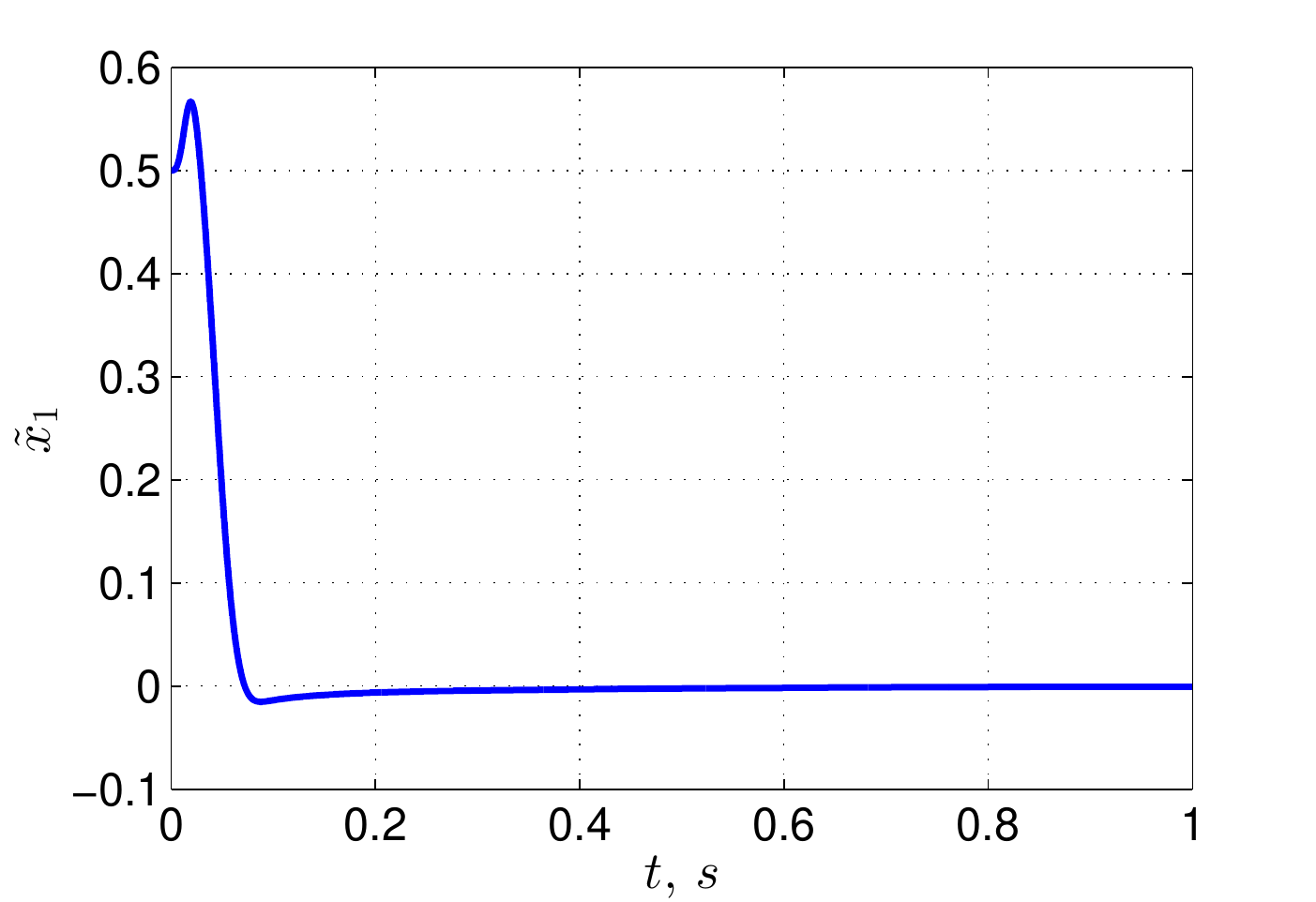}}
\
\subfloat[][]{{\label{fig_3b}}\includegraphics[width=0.43\textwidth]{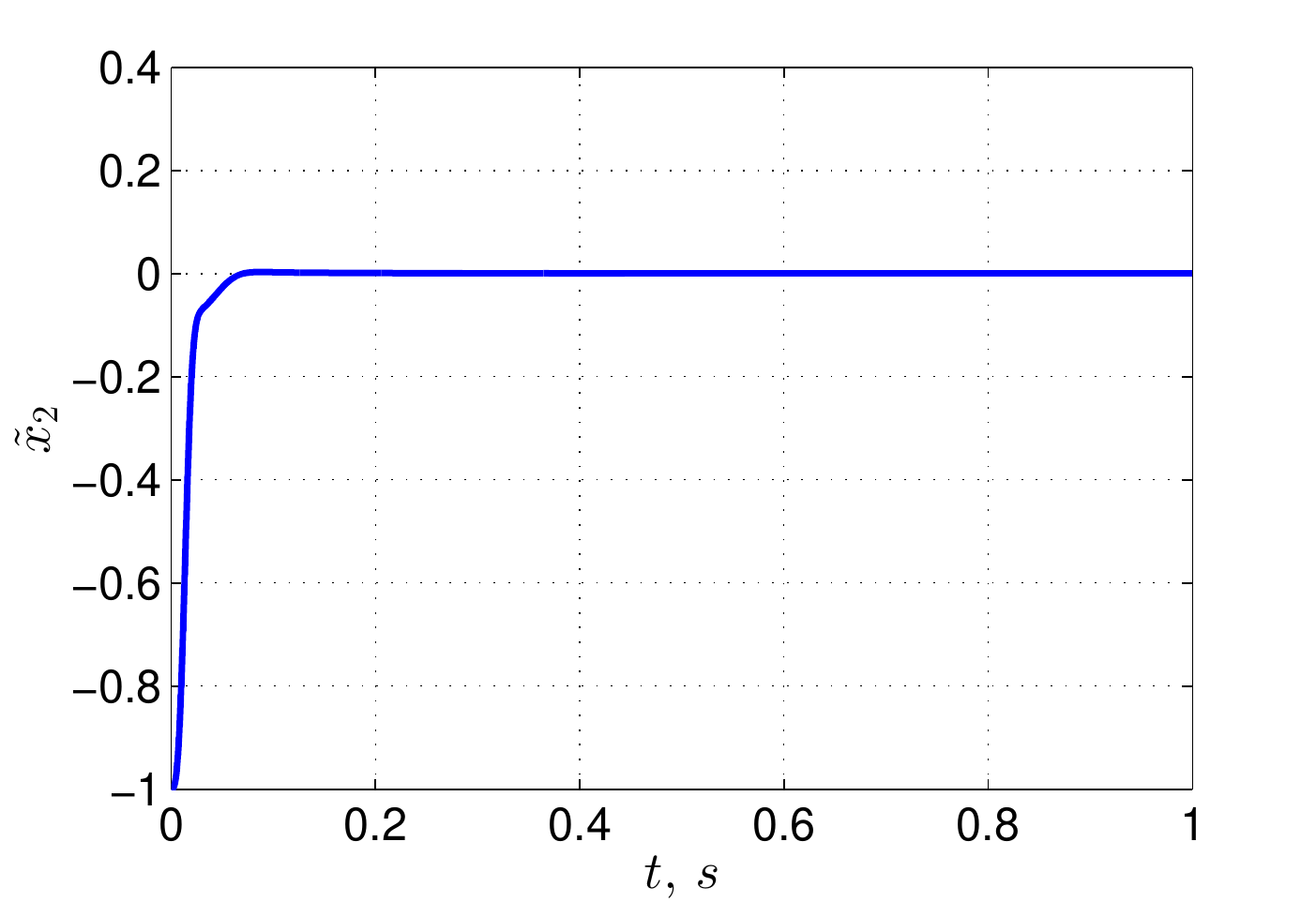}}
\\ \vspace{-3mm}
\subfloat[][]{{\label{fig_3c}}\includegraphics[width=0.43\textwidth]{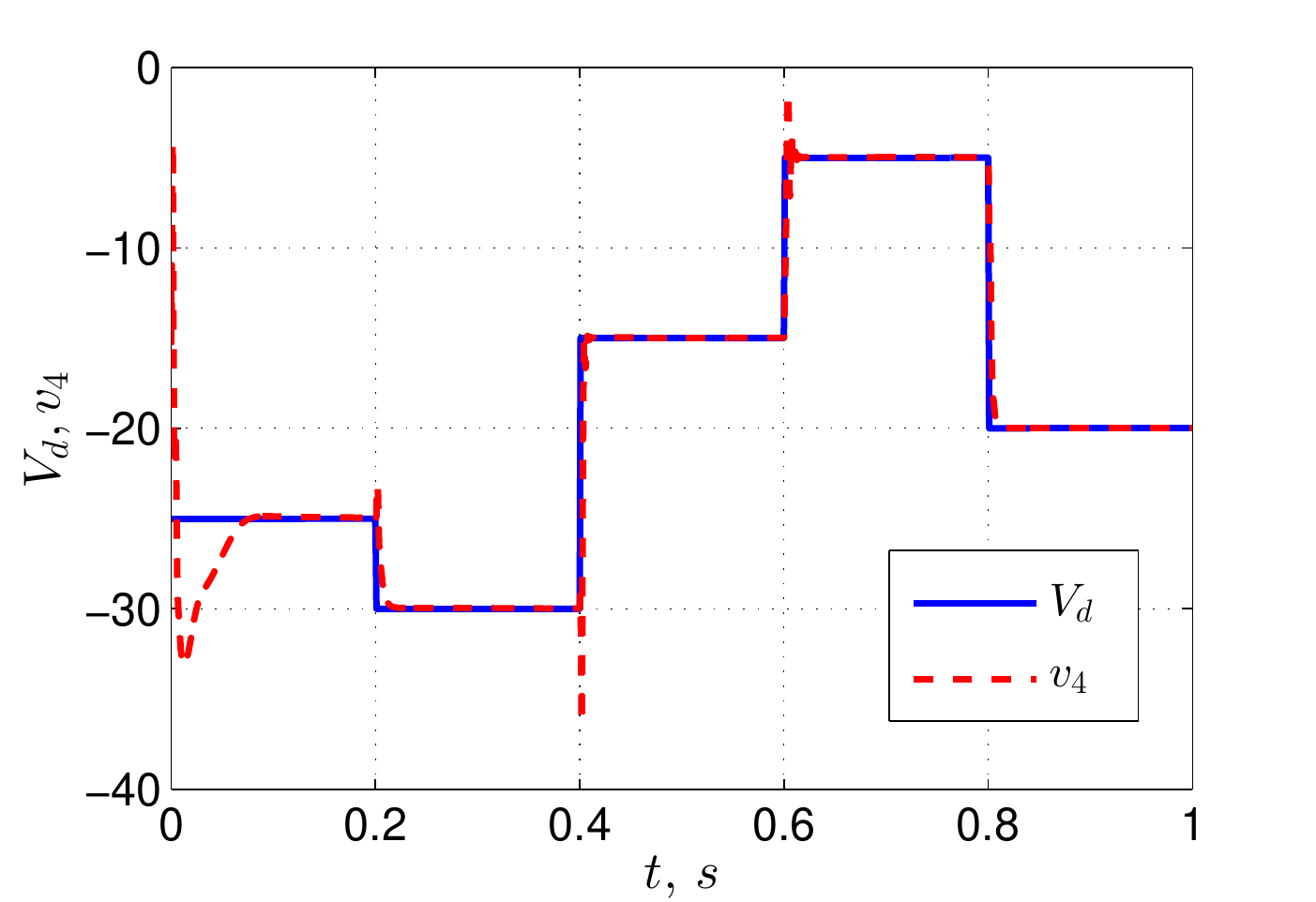}}
\
\subfloat[][]{{\label{fig_3d}}\includegraphics[width=0.43\textwidth]{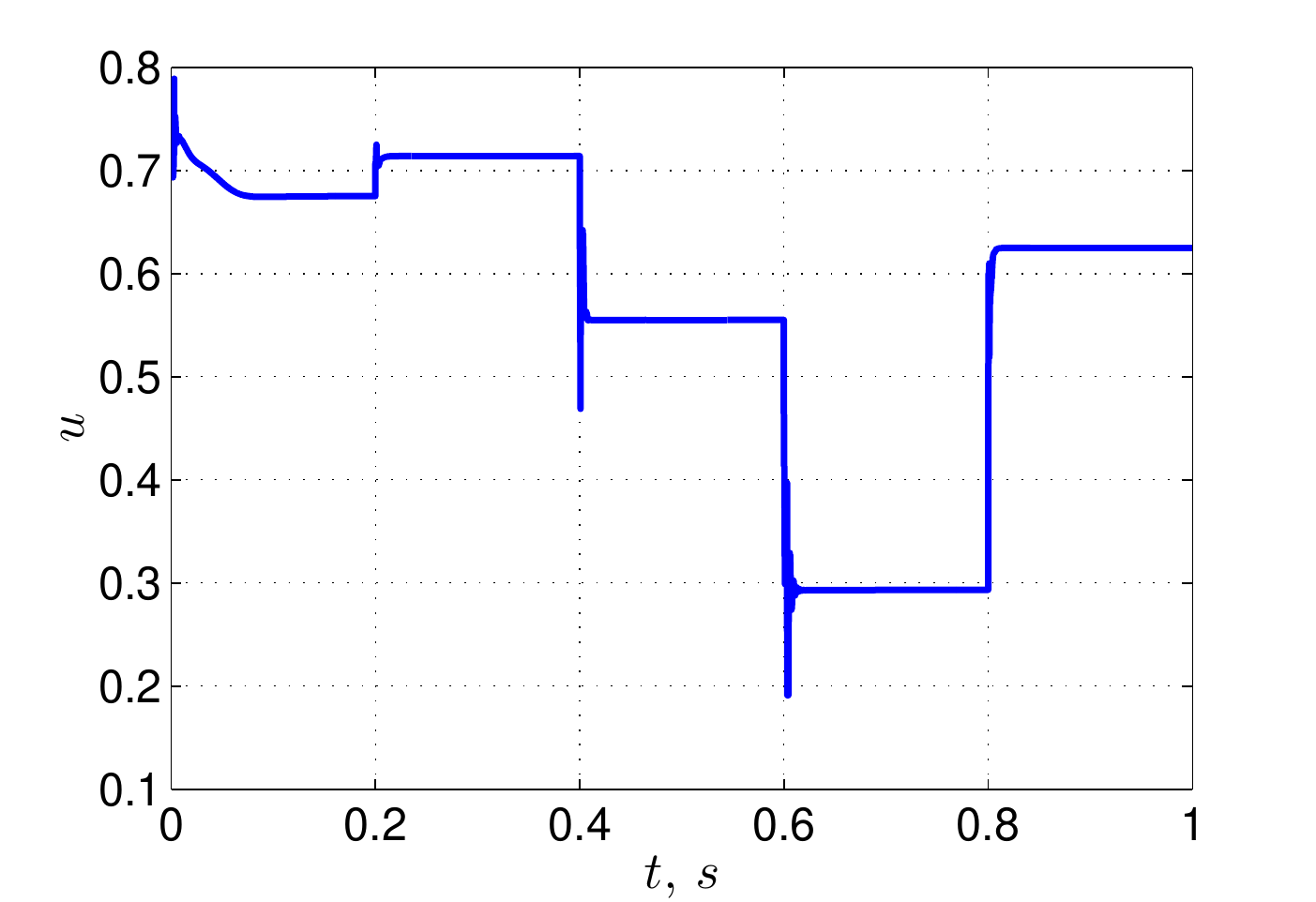}}
\caption{Transients of the observation errors (a) $\tilde x_1:=\hat x_1 - i_1$ , (b) $\tilde x_2:=\hat x_2 - i_3$, (c) the voltage reference   $V_d$ and voltage output $v_4$ and (d) the control input $u$ for  the tuning gains $\alpha=0.5$,  $\Gamma=0.001 I_2$.
}
\label{fig_gbo2}
\end{figure*}

\begin{figure*}[htp]
\centering
\subfloat[][]{{\label{fig_1a}}\includegraphics[width=0.43\textwidth]{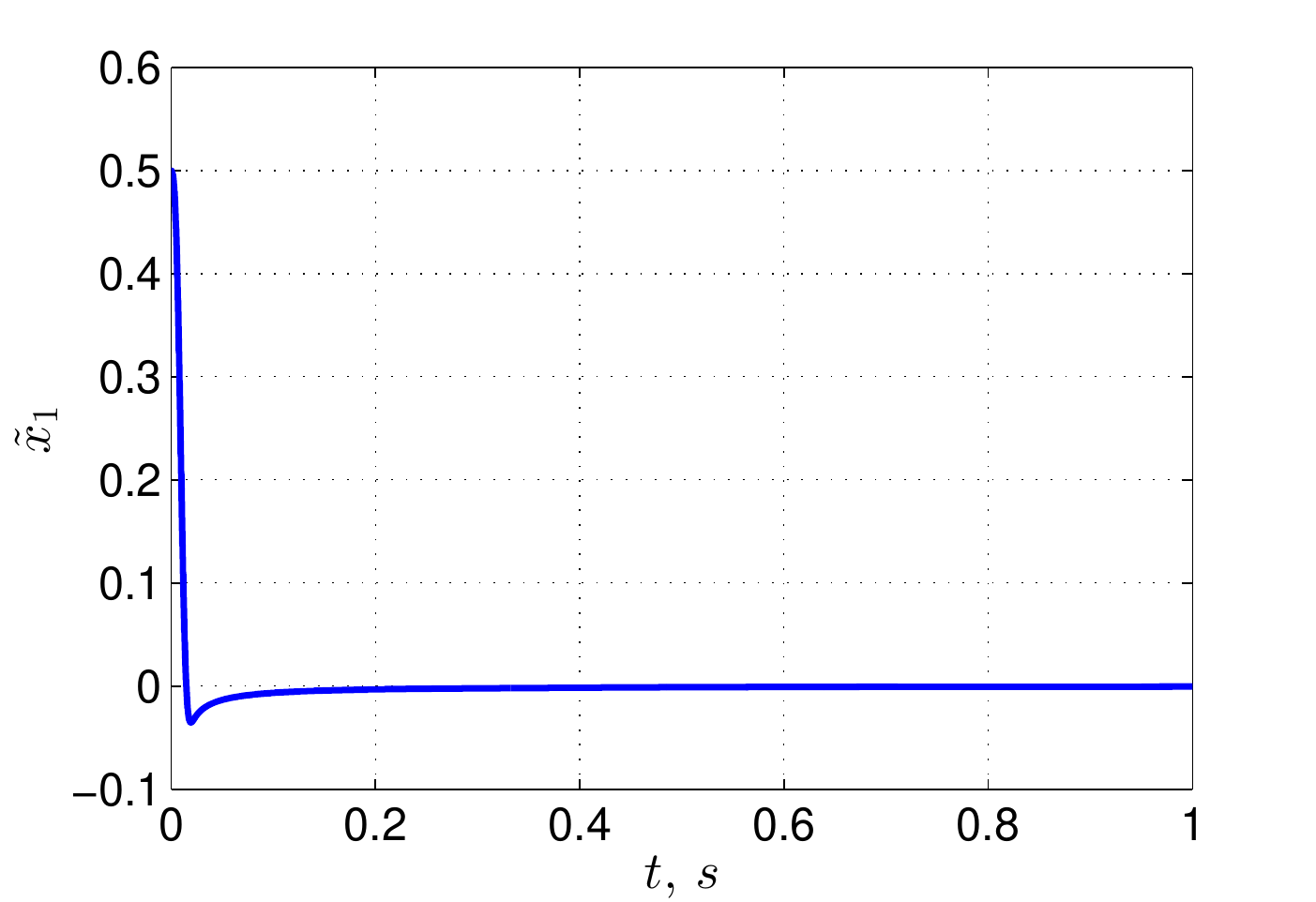}}
\
\subfloat[][]{{\label{fig_1b}}\includegraphics[width=0.43\textwidth]{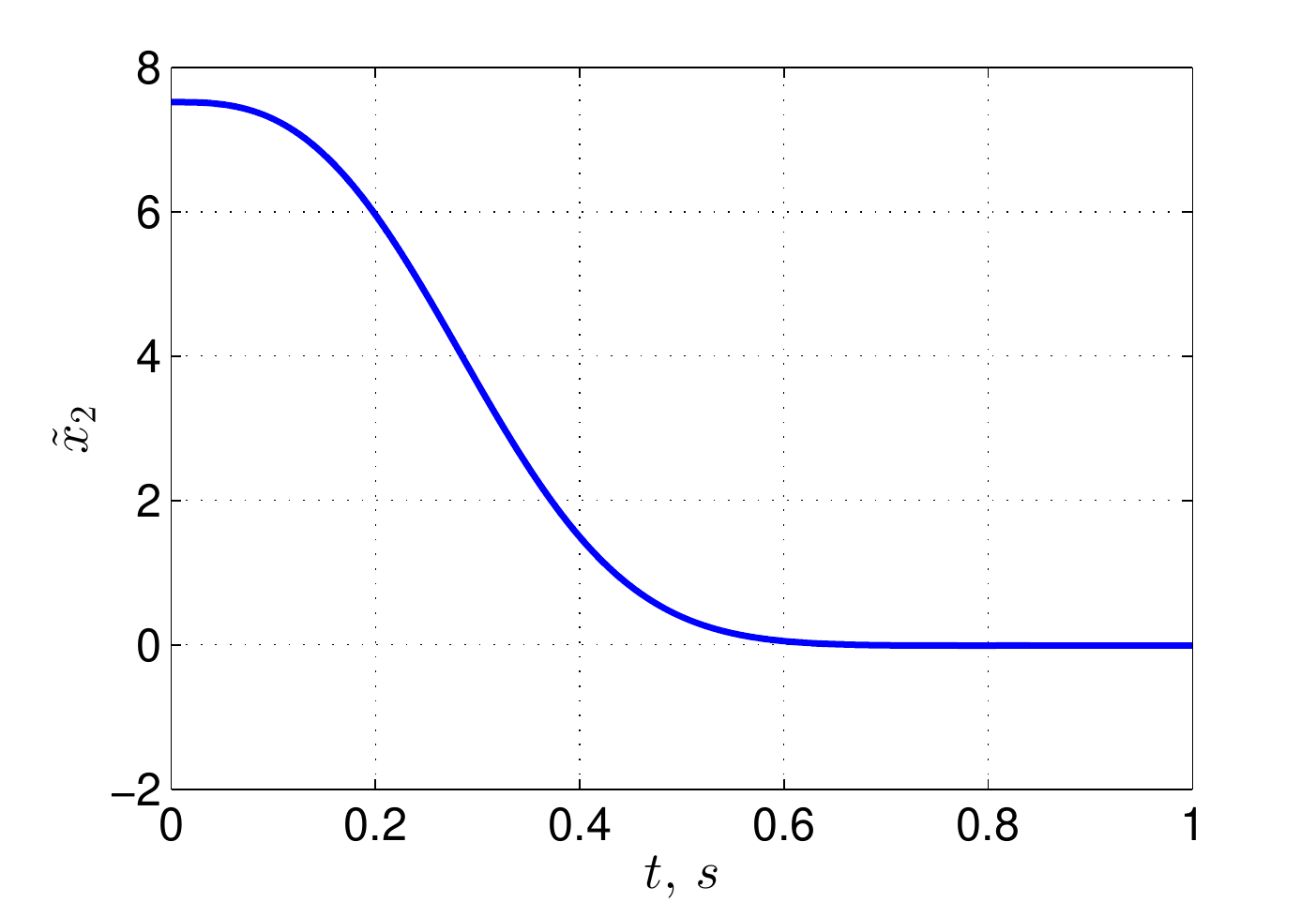}}
\\ \vspace{-3mm}
\subfloat[][]{{\label{fig_1c}}\includegraphics[width=0.43\textwidth]{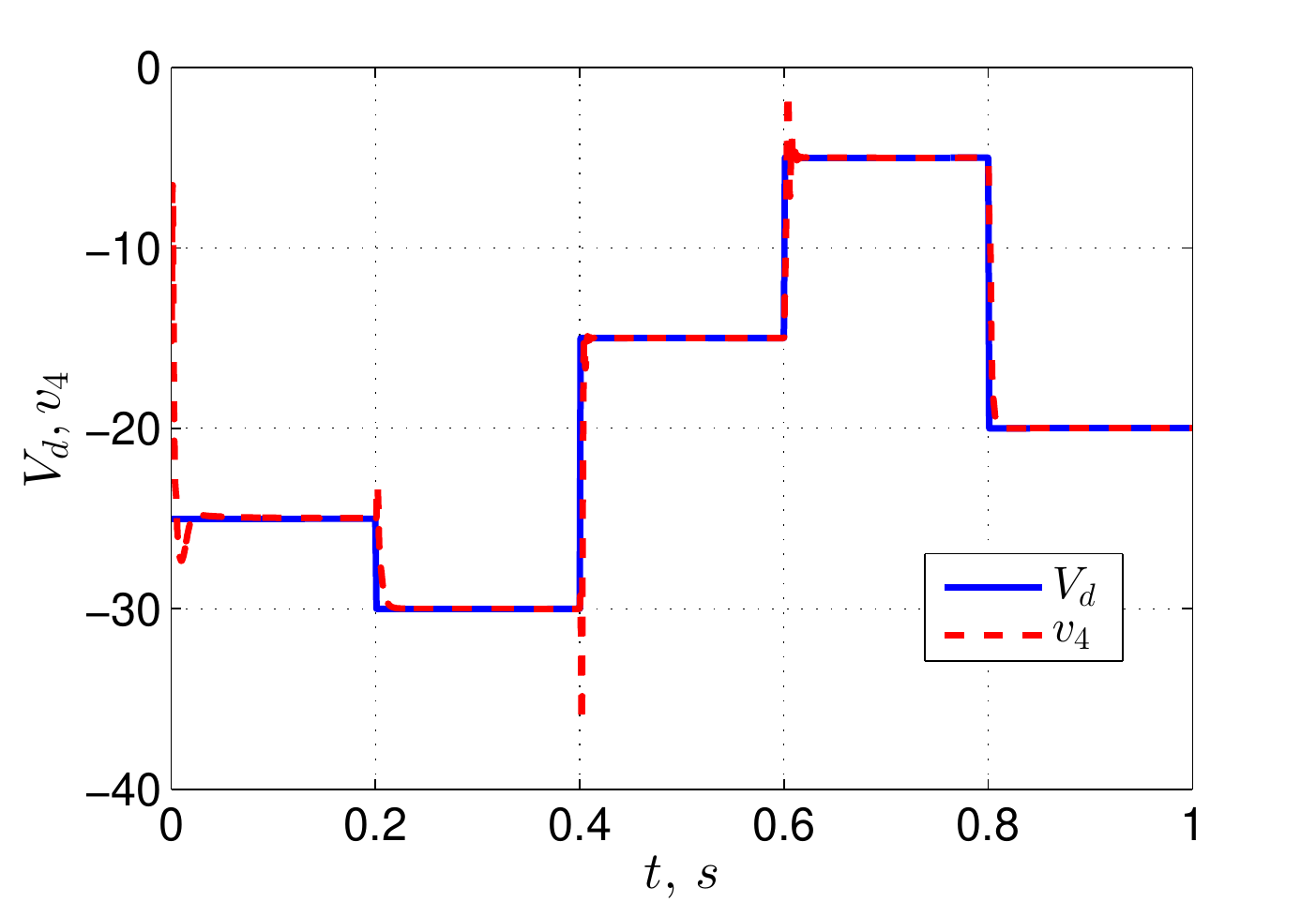}}
\
\subfloat[][]{{\label{fig_1d}}\includegraphics[width=0.43\textwidth]{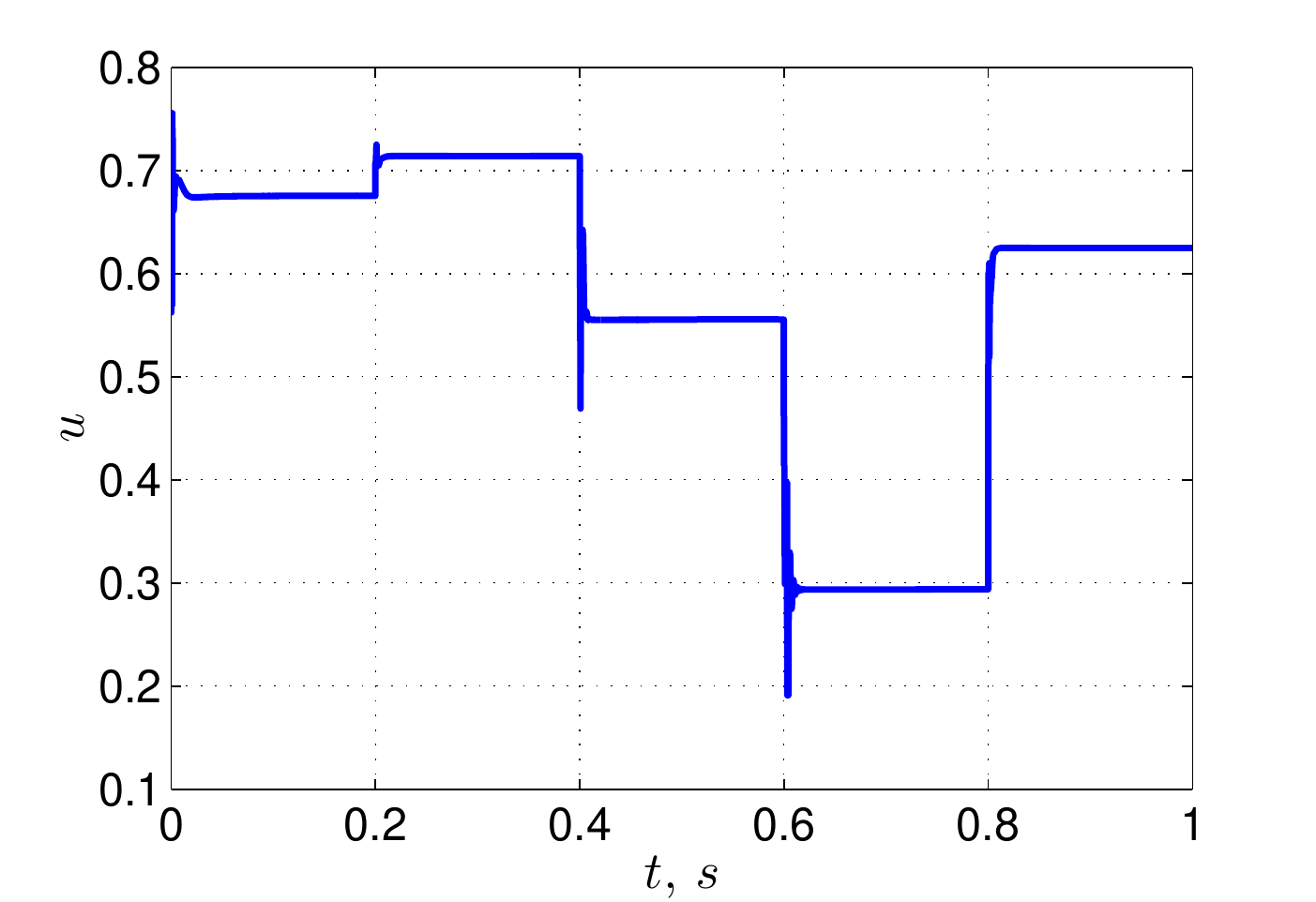}}
\caption{Transients of the observation errors (a) $\tilde x_1:=\hat x_1 - i_1$ , (b) $\tilde x_2:=\hat x_2 - v_4$, (c) the voltage reference   $V_d$ and voltage output $v_4$ and (d) the control input $u$ for  the proposed observer with tuning gains $\alpha=1$,  $\Gamma=\diag\{0.01,0.1\}$.
}
\label{fig_gbo1}
\end{figure*}
}

{
\begin{figure*}[htp]
\centering
\subfloat[][]{{\label{fig_1_2a}}\includegraphics[width=0.43\textwidth]{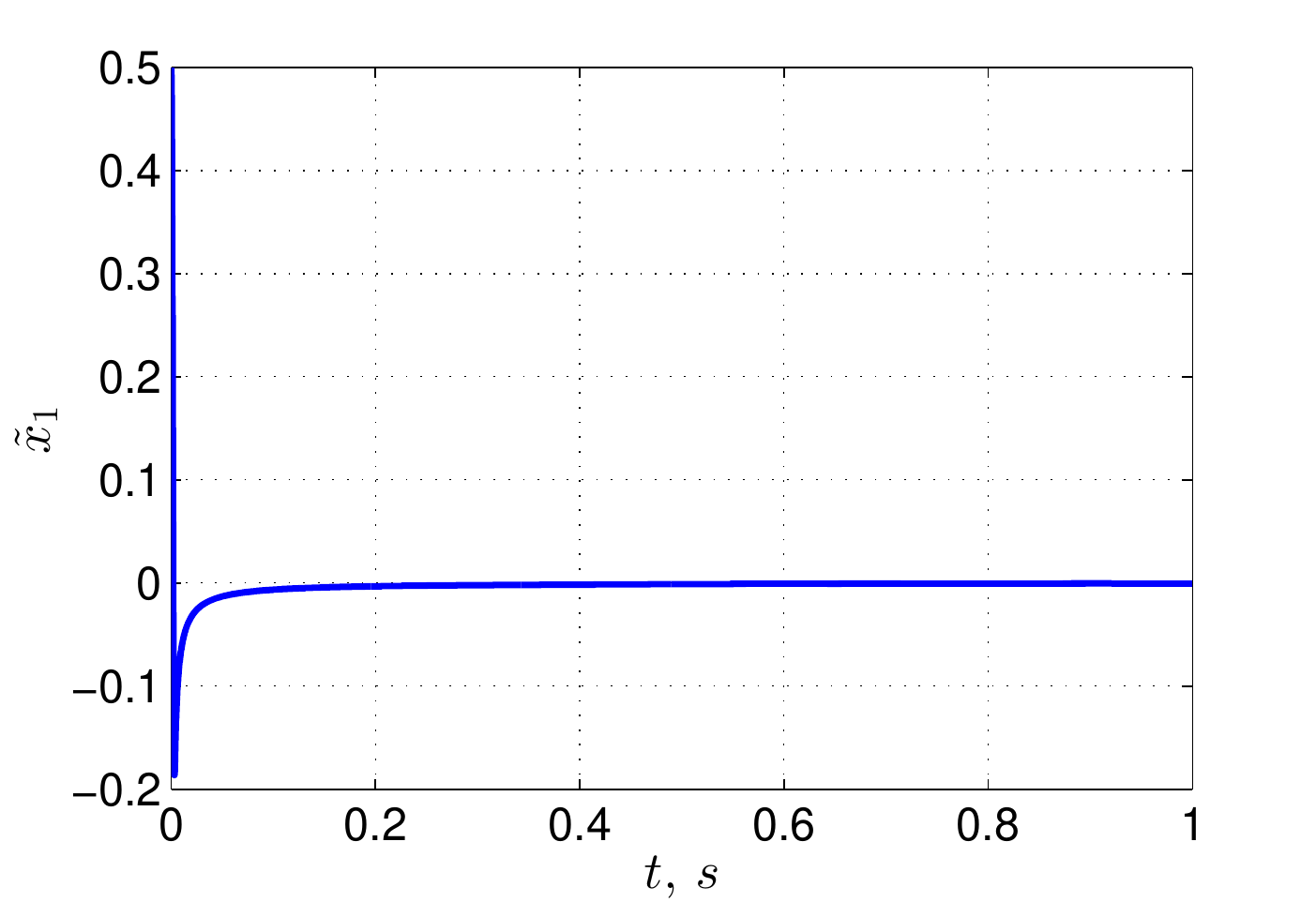}}
\
\subfloat[][]{{\label{fig_1_2b}}\includegraphics[width=0.43\textwidth]{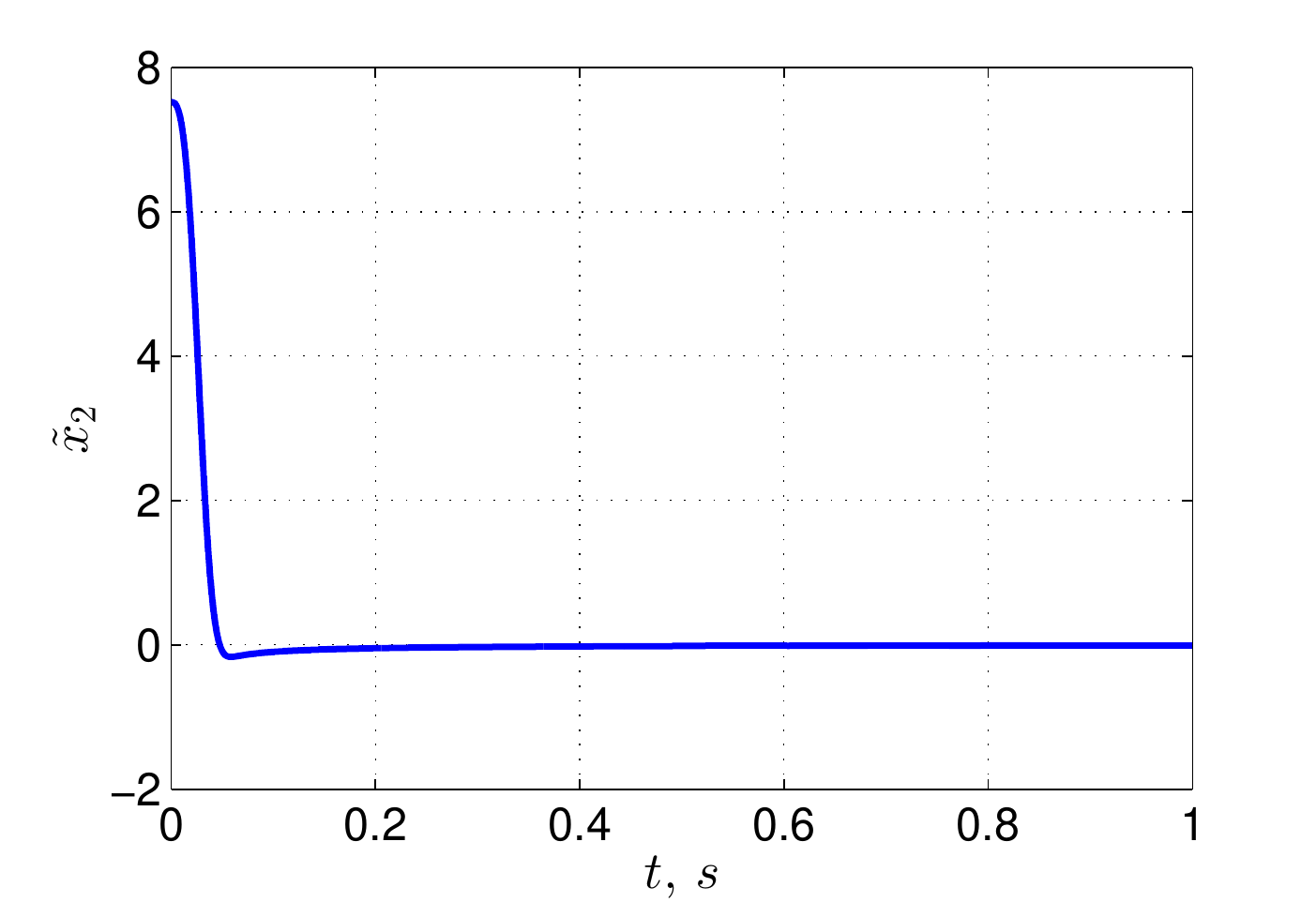}}
\\ \vspace{-3mm}
\subfloat[][]{{\label{fig_1_2c}}\includegraphics[width=0.43\textwidth]{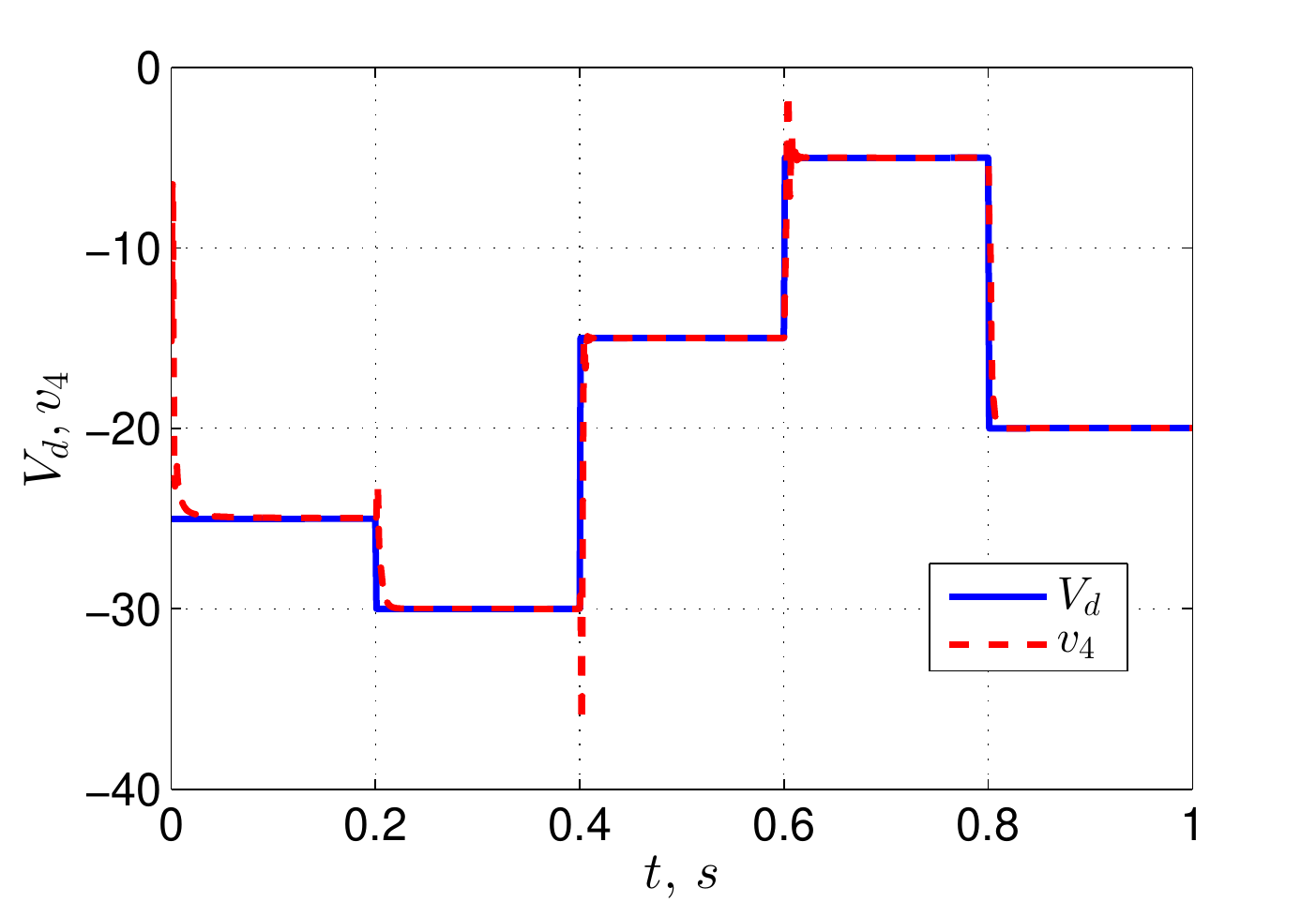}}
\
\subfloat[][]{{\label{fig_1_2d}}\includegraphics[width=0.43\textwidth]{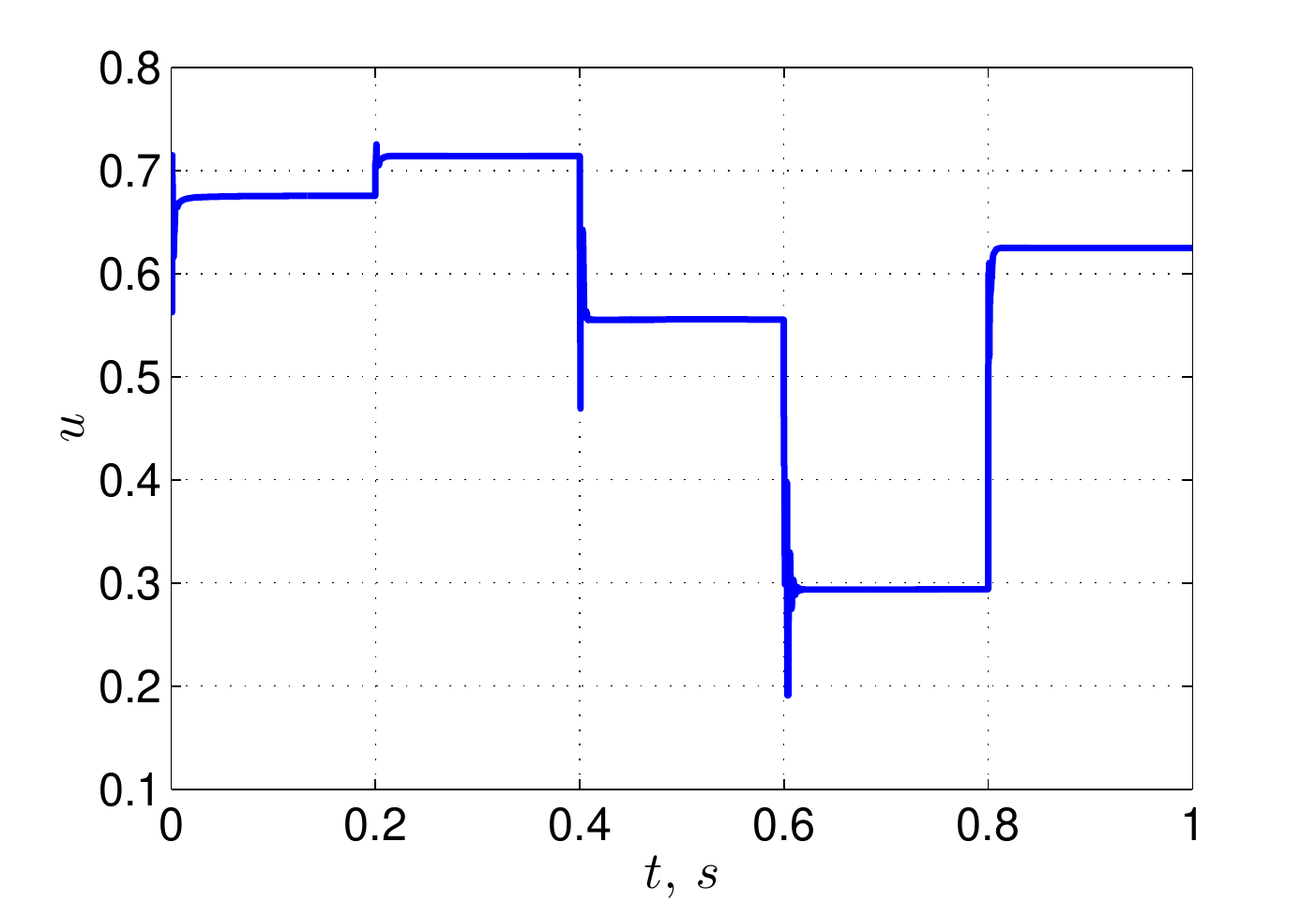}}
\caption{Transients of the observation errors (a) $\tilde x_1:=\hat x_1 - i_1$ , (b) $\tilde x_2:=\hat x_2 - v_4$, (c) the voltage reference   $V_d$ and voltage output $v_4$ and (d) the control input $u$ for  the proposed observer with tuning gains $\alpha=1$,  $\Gamma=\diag\{1,10\}$.
}	
\label{fig_gbo1_2}
\end{figure*}

\begin{figure*}[htp]
\centering
\subfloat[][]{{\label{fig_2a}}\includegraphics[width=0.43\textwidth]{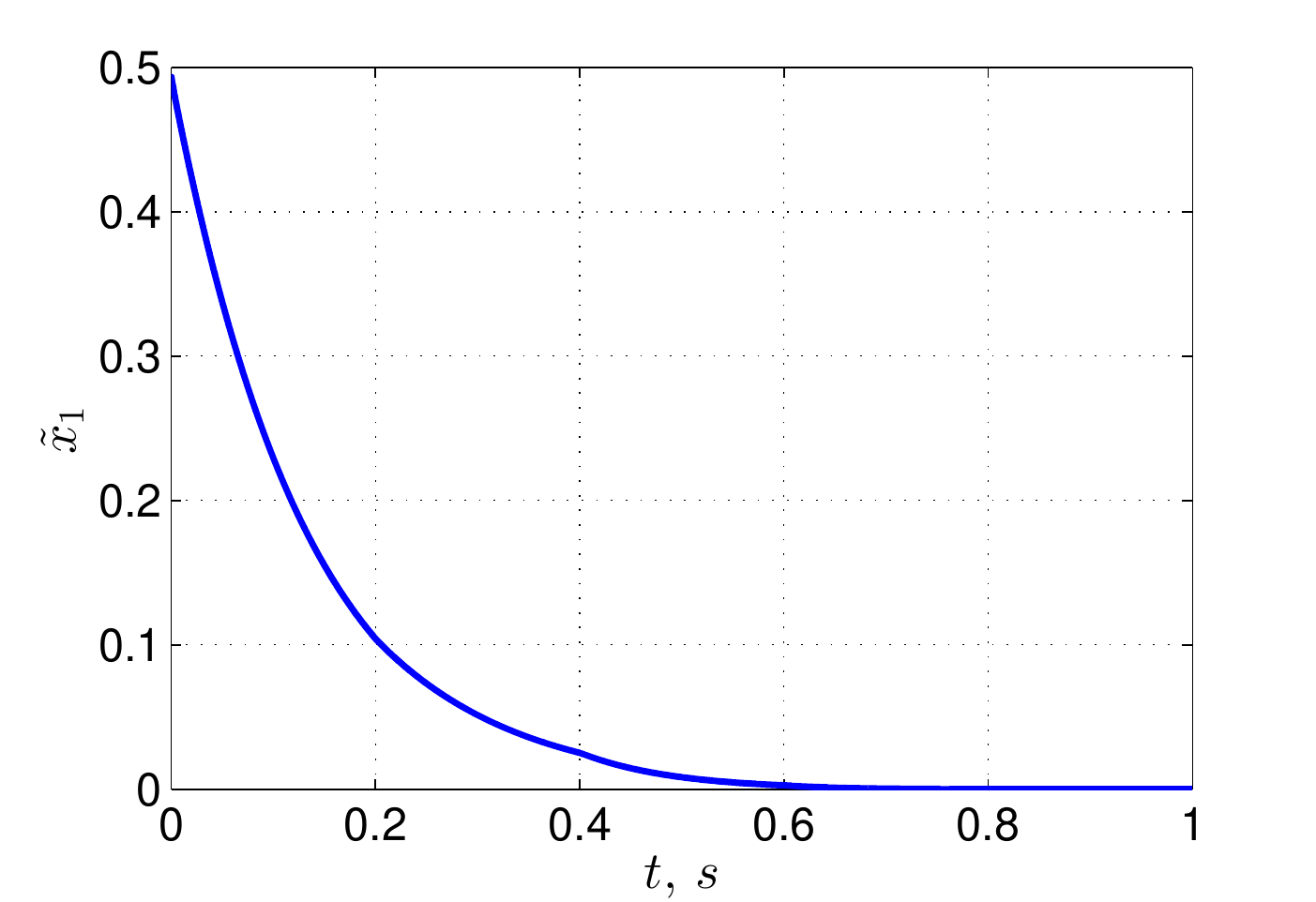}}
\
\subfloat[][]{{\label{fig_2b}}\includegraphics[width=0.43\textwidth]{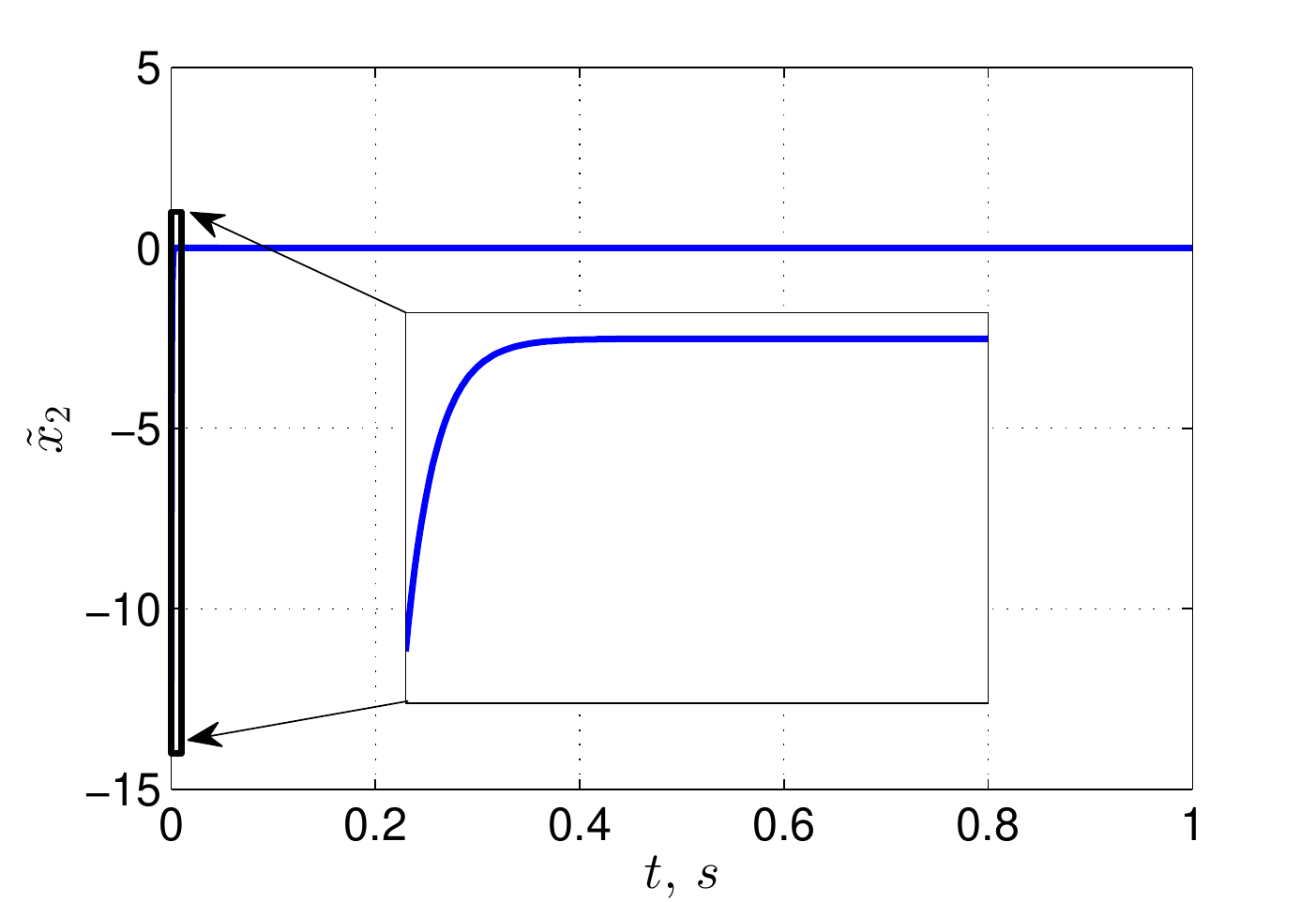}}
\\ \vspace{-3mm}
\subfloat[][]{{\label{fig_2c}}\includegraphics[width=0.43\textwidth]{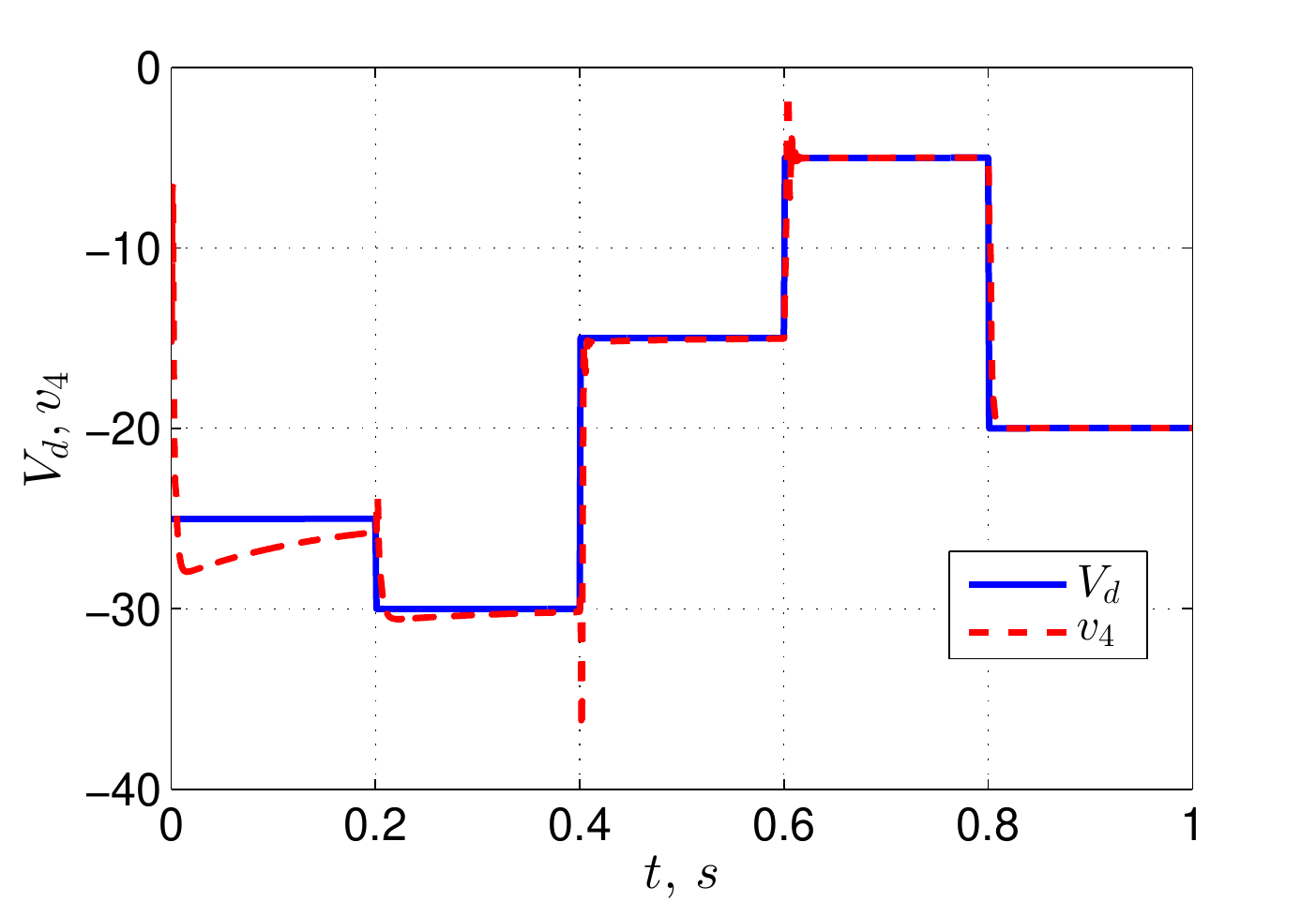}}
\
\subfloat[][]{{\label{fig_2d}}\includegraphics[width=0.43\textwidth]{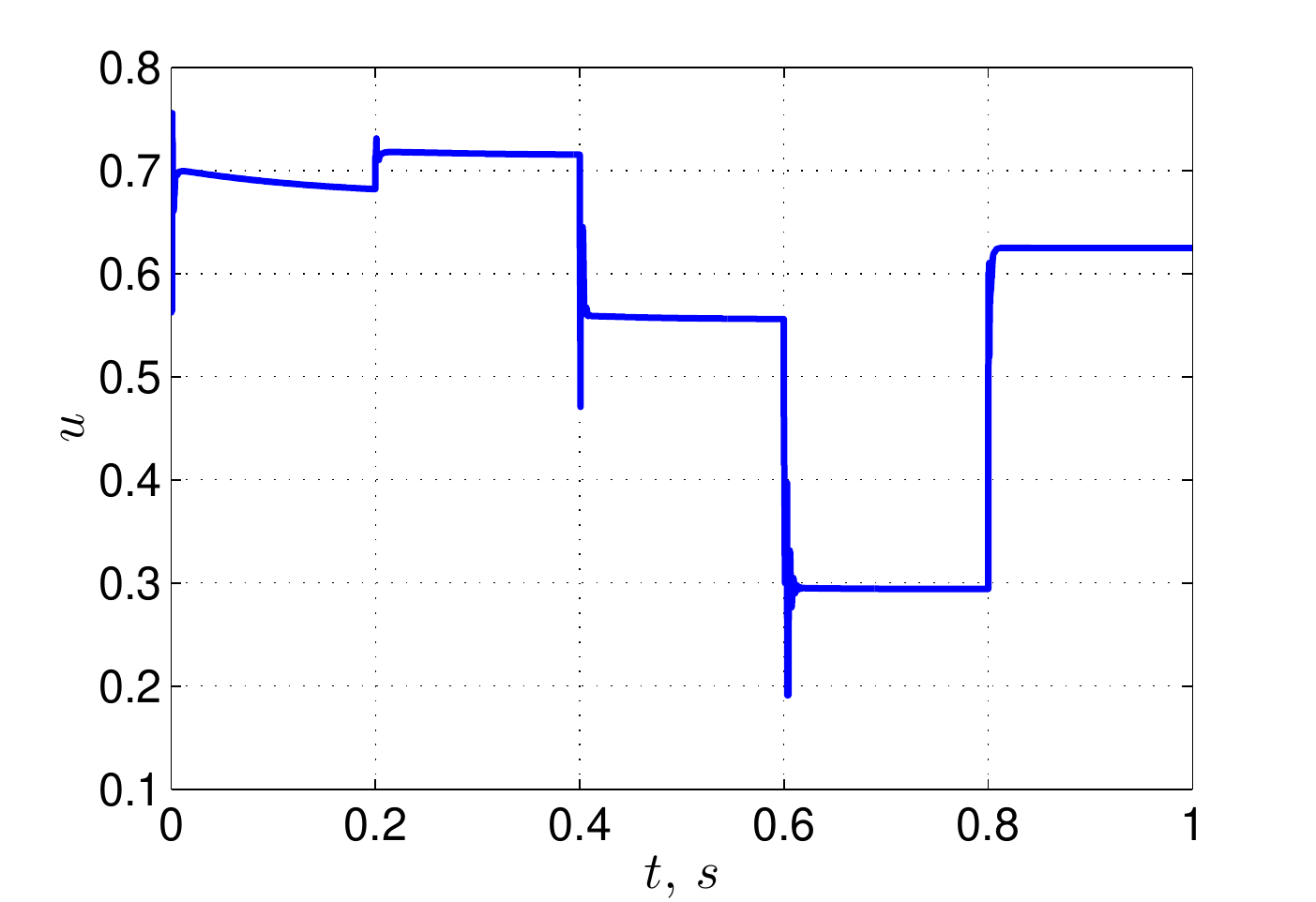}}
\caption{Transients of the observation errors (a) $\tilde x_1$ , (b) $\tilde x_2$, (c) the voltage reference   $V_d$ and voltage output $v_4$ and (d) the control input $u$ for  the I$\&$I observer \eqref{I&I} with tuning gains   $\gamma_1=25$, $\gamma_2=1$.
}
\label{fig_ii1}
\end{figure*}
}

\newpage

\begcen
{\bf \Large Acknowledgement}\\
\endcen
\noindent The first author thanks Laurent Praly for some important clarifications on Assumption \ref{ass1} and Vincent Andrieu for pointing out an error regarding the solvability of the PDE \eqref{pde} found in a previous version of the paper. 

This article is supported by the Ministry of Education and Science of Russian Federation (project 14.Z50.31.0031).

%

\end{document}